\documentclass[11pt]{article}
\usepackage[utf8]{inputenc}
\usepackage[english]{babel}
\usepackage[T1]{fontenc}
\usepackage[margin=1in]{geometry}

\usepackage[]{todonotes}
\usepackage{style-arxiv}
\usepackage{shortcuts-arxiv}

\usepackage[bottom,hang]{footmisc}

\newif\ifdraft
\drafttrue

\begin{document}
	
\begin{center}
	\begin{minipage}[H]{15.5cm} 
		
		\begin{center}
		{\huge \bf Near-Optimal Distributed 2-Ruling Sets on Graphs with Low Arboricity}
		\end{center}
		\vspace{5mm}
		
		\begin{flushleft}
			\hspace{13mm}{\large \textbf{Malte Baumecker}\footnotemark[1], TU Graz -- \href{mailto:malte.baumecker@tugraz.at}{\texttt{malte.baumecker@tugraz.at}}} \vspace{1mm}\\
			\hspace{13mm}{\large \textbf{Rustam Latypov}\footnotemark[2], Aalto University -- \href{mailto:rustam.latypov@aalto.fi}{\texttt{rustam.latypov@aalto.fi}}} \vspace{1mm}\\
			\hspace{13mm}{\large \textbf{Yannic Maus}\footnotemark[1], TU Graz -- \href{mailto:yannic.maus@ist.tugraz.at}{\texttt{yannic.maus@tugraz.at}}} \vspace{1mm}\\
			\hspace{13mm}{\large \textbf{Jara Uitto}, Aalto University -- \href{mailto:jara.uitto@aalto.fi}{\texttt{jara.uitto@aalto.fi}}} \vspace{1mm}\\
		\end{flushleft}
		
		\vspace{5mm}
		
		\begin{abstract}
			Given a graph $G=(V,E)$, a $\beta$-ruling set is a subset of nodes $S\subseteq V$ that is independent, and each node in $V$ is at distance at most $\beta$ from some node in $S$. In this paper, we present almost optimal distributed algorithms for finding $2$-ruling sets in the classical \LOCAL model. Our main contribution is a randomized algorithm that w.h.p.\ computes a $2$-ruling set on any $n$-node graph with bounded arboricity in $O(\log \log n)$ rounds. In fact, the algorithm works up to arboricity $O(\log\log n)$, improves exponentially over the prior state of the art that can be achieved by combining [Barenboim, Elkin, Pettie, Schneider; JACM'16], [Ghaffari; SODA'16], and [Bisht, Kothapalli and Pemmaraju; PODC'14], and nearly matches the lower bound of $\Omega(\log \log n / \log \log \log n)$ [Balliu, Brandt, Kuhn, Olivetti; FOCS'20]. The domination parameter $\beta=2$ is optimal for algorithms with runtime $\log^{o(1)}n$:  on graphs with arboricity $2$, there is a lower bound of $\Omega(\sqrt{\log n})$ rounds for MIS (i.e., $\beta = 1$) [Khoury, Schild; FOCS'25].
			
			\medskip
			Additionally, we obtain improved algorithms for larger arboricity. For general graphs with arboricity $\alpha$, we present a randomized algorithm that computes a $2$-ruling set in $\widetilde{O}(\log^{5/8} \alpha +\log^{5/3} \log n)$ rounds. This improves exponentially over the state of the art for a large range of non-constant arboricity.
			
			\medskip
			Our techniques extend beyond distributed computing. We present an $O(\log \log \log n)$-round algorithm in the low-space Massively Parallel Computation (\mpc) model that w.h.p.\ computes a $2$-ruling set on any graph with arboricity up to $2^{\poly(\log \log n)}$, improving exponentially over the state of the art from [Kothapalli, Pai, Pemmaraju; FSTTCS'20] combined with [Fischer, Giliberti, Grunau; SPAA'23].
		\end{abstract}
		
		\end{minipage}
			
	\thispagestyle{empty}
	\footnotetext[1]{Supported by the Austrian Science Fund (FWF) \url{https://doi.org/10.55776/I6915}. }
	\footnotetext[2]{Supported by the Research Council of Finland, Grant 334238.}
\end{center}

\newpage
\tableofcontents

\newpage
\setcounter{page}{1}
\section{Introduction}

In this paper, we study the problem of finding a $\beta$-ruling set in the standard \LOCAL and \CONGEST models of distributed message-passing.
Given a graph $G = (V, E)$, a $\beta$-ruling set $S \subseteq V$ is a set of non-adjacent nodes such that every node in $G$ is within at most $\beta$ hops from $S$.
This is a natural generalization of the extremely well-studied MIS problem, which corresponds to a $1$-ruling set.
Ruling sets are often faster to compute than MIS and have various applications as subroutines for coloring problems~\cite{GHKM18}, network decomposition~\cite{awerbuch89,GG24}, and MIS itself~\cite{GhaffariImproved16}.
Beyond the applications for other problems, finding ruling sets is a central and well-studied symmetry breaking problem of independent interest~\cite{schneider2013, Balliu2020-ruling, Gfeller07,CONGEST_rulingsets, GG24, Pai2022, kothapalli-superfast2012, cambusCCruling, BEPSv3, schneider2010, assadi-streaming-ruling}.

In general graphs, the classic algorithms developed independently by Luby~\cite{luby86} and Alon, Babai, and Itai~\cite{alon86} can be used to find an MIS in $O(\log n)$ rounds, where $n$ is the number of nodes.
Despite decades of effort, these algorithms remain the state of the art.
For MIS, there is a classic lower bound of $\Omega(\sqrt{\log n / \log \log n})$ rounds~\cite{kuhn16_jacm,10.1145/3519270.3538419}, which was very recently improved to $\Omega(\sqrt{\log n})$ rounds~\cite{khoury2025roundeliminationselfreductionclosing}.
By relaxing the domination to $\beta = 2$, sublogarithmic algorithms are known that run in $O(\sqrt{\log n})$ rounds when expressed as a function of $n$ only~\cite{kothapalli-superfast2012, tushar-super-fast-2014, GhaffariImproved16}.
On the lower bound side, one cannot do better than $\Omega(\log \log n/ (\beta \log \log \log n))$~\cite{Balliu2020-ruling}, leaving an exponential gap between lower and upper bounds, e.g., for $\beta=2$ on general graphs.

\medskip

\begin{tcolorbox}
In this paper, we close this gap on sparse graphs up to triple-logarithmic
factors, giving exponentially faster $2$-ruling-set algorithms for graphs of
low arboricity.
\end{tcolorbox}
\medskip

Arboricity is a standard measure of sparsity: it corresponds to the number of forests that are needed to partition the edge set. Symmetry-breaking in graphs with bounded arboricity is frequently studied in \LOCAL, \CONGEST, dynamic models, and various other models of computation~\cite{RotenbergImprovedDynColouringSparse,GhaGrunauDynArbColoring,EdenMosselRonApxArbSublinear,CCHIQRSAdapditveOutOrient,ERSsublinearApxkCliquesLowArb, cambus_et_al:LIPIcs.DISC.2021.15,morgan_et_al:LIPIcs.DISC.2021.33,faour_et_al:LIPIcs.DISC.2025.54,pai_et_al:LIPIcs.DISC.2017.38,kothapalli-superfast2012,BarenboimElkinKuhnDistLinearColoring,Dory_Ghaffari_Ilchi_2023,LenzenWattenhoferMinDomSetInBoundedArb}. 

At a high level, our main technical contribution is a fast degree drop
procedure. Given a graph of maximum degree $\Delta$, the goal of our procedure is to compute an independent set $S$ so that, after removing all nodes covered by $S$, the remaining graph has maximum degree significantly smaller than $\Delta$.
Given that we require independence of $S$ and the coverage of all high-degree nodes, the probability of including a node $u$ to $S$ necessarily depends on its local graph topology. Hence, short cycles create dependencies on the events that nearby nodes are selected and standard tools such as Chernoff bounds cannot be used to obtain concentration.
We control these dependencies using read-$k$
concentration inequalities, building on techniques originally introduced by
Pemmaraju and Riaz for faster MIS algorithms on bounded-arboricity
graphs~\cite{pemmaraju_et_al:LIPIcs.OPODIS.2016.9}. We discuss our approach and
its relation to their techniques in \Cref{subsec:tech_over}. An advantage of our degree drop is that it extends naturally to the memory-restricted \emph{low-space} Massively Parallel Computation (\mpc) model, where ruling sets have also been extensively studied~\cite{GU19,Ghaffari2020,Kothapalli2020,CzumajDaviesParter20,GilibertiParsaeian24,FischerGilibertiGrunau23,JiKothappalliPemmarajuSingh25}. There, we design an $O(\log \log \log n)$-round algorithm that is simple, yet improves exponentially on the state of the art. 



\subparagraph{Previously Known Faster Algorithms in Sparse Graphs via Degree Drop.}
The lower bounds of $\Omega(\sqrt{\log n / \log \log n})$ for MIS and
$\Omega(\log \log n/(\beta \log \log \log n))$ for $\beta$-ruling sets already
hold on trees~\cite{10.1145/3519270.3538419,Balliu2020-ruling}. Nevertheless,
sparse graphs often admit faster algorithms than general graphs. A central
reason is that sparsity can be exploited to reduce the maximum degree quickly.

A seminal example is the degree-reduction framework for MIS of Barenboim, Elkin,
Pettie, and Schneider~\cite{BEPSv3}. For graphs of arboricity at most $\alpha$,
their tool reduces the maximum degree to
$\alpha \cdot 2^{O(\sqrt{\log n})}$ in $O(\sqrt{\log n})$ rounds. The analysis relies
on reasoning about nodes that are mutually at distance at least $3$; this separation
ensures the needed independence between selected nodes, but it also limits the
rate at which the degree can drop. More recently, Khoury and Schild improved
this type of degree reduction on trees, obtaining an
$o(\sqrt{\log n})$-round MIS algorithm and thereby separating trees from
general graphs for MIS (in fact they already separate trees from graphs with arboricity two)~\cite{khoury2025breakingbarriersdistributedmis}. Their
algorithm is based on a two-stage competition process in which each node
survives with subconstant probability.

These degree-reduction tools, however, are designed for MIS, and their
runtime remains constrained by MIS-type lower bounds. Indeed, even combining the
degree reduction of~\cite{BEPSv3} with the ruling-set subsampling of Bisht,
Kothapalli, and Pemmaraju~\cite{tushar-super-fast-2014} and Ghaffari's
efficient algorithm for low-degree graphs~\cite{GhaffariImproved16} yields only
a $\log^{\Omega(1)} n$-round algorithm for $2$-ruling sets on
bounded-arboricity graphs; see \Cref{apx:related_work} for more details. To obtain
double-logarithmic runtimes for $2$-ruling sets, one needs a faster degree drop.
On trees, this is possible~\cite{BMU25}: roughly speaking, the maximum degree can be reduced
from $\Delta$ to $\sqrt{\Delta}$ in $O(1)$ rounds, yielding an exponential
separation between MIS and $2$-ruling sets.


The difficulty is that this separation currently relies on the input graph
being locally tree-like. The same issue appears in the fast MIS degree-reduction
tools above, which rely on the absence of short cycles, or at least on enough
local independence induced by distance-$3$ separation in the analysis. This is an
obstacle for bounded-arboricity graphs: although they are globally sparse, they
may contain many short cycles. Consequently, existing techniques do not seem to
provide the constant-round polynomial degree drop needed to close the
exponential gap between upper and lower bounds for $2$-ruling sets.

\subparagraph{Avoiding Dependencies with $\beta \geq 3$.}
One way to bypass the dependencies created by short cycles is to increase the
domination radius. Kothapalli and Pemmaraju~\cite{kothapalli-superfast2012}
show how to sample low-degree subgraphs so that every high-degree node in the
original graph is within distance $2$ of a sampled node. Computing MIS in
these sampled graphs then yields an $O(\log^3 \log n)$-round algorithm for
$3$-ruling sets for bouned arboricity graphs.

This approach, however, does not seem to extend to $2$-ruling sets. For
$2$-ruling sets, the sampled nodes would need to cover the immediate
neighborhood of every high-degree node, rather than only its distance-$2$
neighborhood. Sampling sufficient low-degree subgraphs with this stronger
coverage property appears difficult. Thus, avoiding dependencies by increasing
the domination radius does not resolve the case where $\beta=2$.

\subsection{Our Contributions}

We consider the setting where a communication network is modeled as a graph $G=(V,E)$, where $|V|=n$ nodes with unbounded computational power communicate over the edges in synchronous message-passing rounds. In the distributed \LOCAL model, introduced by Linial \cite{linial92}, message sizes are unbounded, whereas in the \CONGEST~\cite{Peleg1987DistributedCA} model messages are bounded by $O(\log n)$ bits. In the beginning, each node only knows its neighbors, and upon termination, each node $v$ should know its own output, for example, whether it is part of the ruling set or not. The runtime is measured in the number of rounds. We consider graphs of low arboricity, where the arboricity\footnote{Formally, the arboricity of $G$ is defined as $\alpha(G)\coloneq \max_{G'\subseteq G,|V(G')|\geq2} \big\{\big\lceil \frac{|E(G')|}{|V(G')|-1 }\big\rceil\big\}$.} of a graph $G$ is defined as the minimum number of forests into which the edges of $G$ can be partitioned.

First, we present a fast algorithm for ruling sets with the best possible domination distance $\beta=2$ for subpolylogarithmic runtime on graphs with low arboricity.

\begin{restatable}{theorem}{thmRSloglogArb}
    \label{thm:2RSconstArb}
        There is an $O(\log\log n)$-round \CONGEST algorithm that w.h.p.\ computes a $2$-ruling set in graphs with arboricity up to $O(\log \log n)$.
\end{restatable}

Our algorithm is optimal up to a triple-logarithmic factor and improves exponentially compared to the fastest known $2$-ruling set algorithm that can be deduced from published algorithms and runs in~$O(\log^{1/4} \Delta+\poly \log \log n)$ rounds~\cite{BEPSv3,tushar-super-fast-2014,GhaffariImproved16}, see \Cref{cor:related_work} for details. 
Prior to our work, the only double-logarithmic–time ruling set algorithm for graphs with arboricity $\alpha > 1$ achieved an $O(\log^3 \log n)$ runtime and computed a $3$-ruling set for graphs of arboricity $O(1)$ \cite{kothapalli-superfast2012}, see~\Cref{fig:map}. This exponential gap between $3$-ruling sets and $2$-ruling sets highlights a qualitative increase in difficulty: even modest improvements in the ruling parameter have required substantially higher runtimes. In contrast, compared to the existing $3$-ruling set algorithm, our result achieves a polynomially faster runtime while solving the much harder $2$-ruling set problem on a broader graph class.

\begin{figure}
	\centering
	\includegraphics[width=0.9\textwidth]{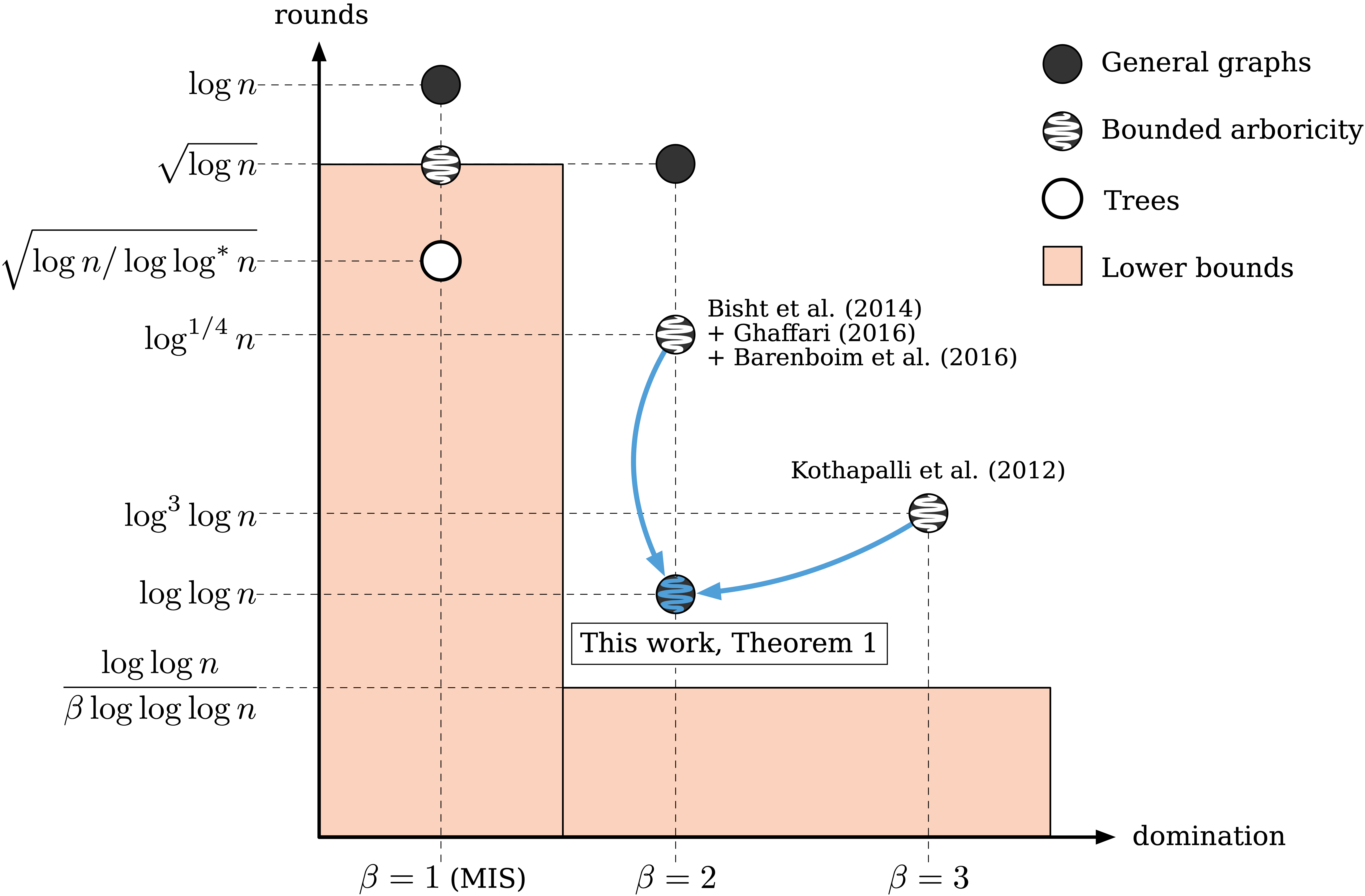}
    	\caption{Distributed algorithms for ruling sets: upper bounds are shown as dots (fill pattern indicates graph class) and lower bounds as grey regions. The horizontal axis represents the domination distance (1-ruling set (MIS), 2-ruling set, and 3-ruling set), and the vertical axis represents the round complexity as a function of $n$ (with $\Delta$ expressed as a polynomial in $n$). Blue arrows indicate improvements over the prior state of the art: our algorithm is exponentially faster for 2-ruling sets, and polynomially faster for 3-ruling sets while simultaneously achieving the optimal domination distance for subpolylogarithmic runtimes. Our runtime is optimal up to a triple-logarithmic factor. Prior to this work, only trees (arboricity 1) admitted an $O(\log \log n)$-round algorithm~\cite{BMU25}.}
	\label{fig:map}
\end{figure}




Our algorithm runs in the \CONGEST model and thus \Cref{thm:2RSconstArb} improves upon the best known \CONGEST algorithm for $2$-ruling sets that runs in~$O(\log\Delta \cdot (\log n)^{1/2+\epsilon}+\frac{\log n}{\log \log n})$ rounds \cite{pai_et_al:LIPIcs.DISC.2017.38}.

\medskip

Second, we present an algorithm whose runtime remains sublogarithmic for graphs with any arboricity.

\begin{restatable}{theorem}{thmMORE}
    \label{thm:MORE}
    There is an $\widetilde{O}(\log^{5/8} \alpha+\log^{5/3} \log n)$-round \LOCAL algorithm to compute a $2$-ruling set in graphs with arboricity $\alpha$ w.h.p.\ 
\end{restatable}

To the best of our knowledge, the best $2$-ruling set algorithm for graphs with larger arboricity follows by combining prior work \cite{GhaffariImproved16,BEPSv3, tushar-super-fast-2014} and runs in $O(\log^{1/4}\Delta+\log\alpha+\poly \log \log n)$ rounds; see~\Cref{thm:related_work} for details. 

\smallskip 

Finally, we observe that \Cref{thm:MORE} runs in $\poly \log \log n$ rounds for graphs with arboricity up to $2^{(\log \log n)^{8/3}}$.

\begin{restatable}{corollary}{corARB}
    \label{cor:ARB}
     There is an $\widetilde{O}(\log^{5/3} \log n)$-round \LOCAL algorithm to compute a $2$-ruling set in graphs with arboricity $\alpha \leq 2^{(\log \log n)^{8/3}}$ w.h.p.\
\end{restatable}

\Cref{cor:ARB} improves upon the state of the art beyond the $\log \log n$ arboricity barrier of~\Cref{thm:2RSconstArb}. In fact, \Cref{thm:MORE} still improves polynomially on the state of the art $2$-ruling set algorithm up to arboricity $2^{(\log n)^{4/5}}$.

\subparagraph*{Massively Parallel Computation.}
The Massively Parallel Computation (\mpc) model~\cite{KarloffSV10,BeameKoutrisSuciu17} is one of the leading theoretical frameworks for large-scale computation, striking a balance between mathematical abstraction and practical relevance. It faithfully captures the behavior of widely deployed data processing frameworks such as Hadoop~\cite{White2009}, Spark~\cite{Zaharia2010}, MapReduce~\cite{Dean2008}, and Dryad~\cite{Isard2007}. In the \emph{low-space} variant of the \mpc model there are $M$ machines, each with \emph{local space} $S$, sublinear in the input size. For graphs, the edges of the input graph $G=(V, E)$ are arbitrarily and equally distributed among the machines, each with local space $S=O(n^{\eps})$ such that $\eps \in (0,1)$. Within a round, machines perform all-to-all synchronous communication, and then perform some unbounded computation on local data. The system is fault-free, meaning that nodes never crash and messages never get corrupted. The \emph{total space} refers to the maximum amount of memory used by the collection of $M$ machines at any point during the algorithm.

Ruling sets have been extensively studied in the low-space \mpc model in recent years \cite{GU19,Ghaffari2020,Kothapalli2020,CzumajDaviesParter20,GilibertiParsaeian24,FischerGilibertiGrunau23,JiKothappalliPemmarajuSingh25}. To the best of our knowledge, for the broad class of graphs with arboricity $\alpha \leq 2^{\poly (\log \log n)}$, the current state-of-the-art runtime for 2-ruling set is $\poly(\log \log n)$. While not published, this result follows by combining the algorithm of \cite{Kothapalli2020} with the degree reduction by \cite{FischerGilibertiGrunau23}. We improve upon this exponentially with a sharp triple-logarithmic 2-ruling set algorithm, using the same amount of total space.


\begin{restatable}{theorem}{thmRSloglogArbMPC}
	\label{thm:2RSconstArbMPC}
	There is an $O(\log\log\log n)$-round low-space \mpc algorithm that w.h.p.\ computes a $2$-ruling set in graphs with arboricity $\alpha \leq 2^{\poly (\log \log n)}$, using $O(m + n^{1+\eps})$ total space.
\end{restatable}

\subsection{Related Work}
We have already covered parts of the related work while presenting our results, but we want to present a more thorough overview. To complement the \LOCAL overview given here, we refer to~\Cref{tab:rand-rul-sets-extended} for the randomized regime and~\Cref{tab:det-rul-sets} for the deterministic regime in~\Cref{apx:related_work}.

\subparagraph*{Randomized \LOCAL.} Ghaffari's seminal MIS algorithm is still the state of the art for MIS on general graphs and runs in $O(\log\Delta+\poly\log\log n)$ rounds~\cite{GhaffariImproved16}, which improves on Luby's $O(\log n)$\cite{luby86, alon86} algorithm for $\Delta=n^{o(1)}$.  His algorithm can also be combined with prior work on ruling sets \cite{tushar-super-fast-2014} to compute $\beta$-ruling sets in $O(\beta\log^{1/\beta}\Delta+\poly\log\log n)$ rounds.

Recently, there was progress for MIS and the conjecture of a $\Theta(\sqrt{\log n})$ complexity for MIS on trees was refuted by Khoury and Schild, as they showed that an MIS can be computed in $O(\sqrt{\log n/ \log \log^*n})$ rounds on trees, and in $O({\frac{\log \Delta}{\log \log^* \Delta}}+\poly (\log \log n))$ rounds on graphs of girth at least $7$ \cite{khoury2025breakingbarriersdistributedmis}. 
 
Additionally, they gave an $\Omega(\min\{\log \Delta,\sqrt{\log n}\})$ lower bound for Maximal Matching, and thus for MIS on line graphs, improving on the previous known lower bounds and showing that MIS on general graphs is strictly harder than on trees. 
For general graphs, there has been a series of works started by an~$\Omega(\min\{\sqrt{\log n/\log \log n},\log \Delta/\log \log \Delta\})$ lower bound on line graphs by Kuhn, Moscibroda, and Wattenhofer~\cite{kuhn16_jacm}, for which Coupette and Lenzen presented a simple proof~\cite{coupetteLenzen2021BreezingKMW}. This was followed by multiple works of Balliu, Ghaffari, Kuhn, and Olivetti, with the latest paper extending the lower bound to trees~\cite{10.1145/3519270.3538419,BBKO2021hideandseek,BBO22}.
For graphs of arboricity $\alpha$, Barenboim, Elkin, Pettie and Schneider showed that it is possible to compute an MIS in $O(\log\alpha + \sqrt{\log n})$ rounds \cite{BEPSv3,GhaffariImproved16}, which results in an $O(\sqrt{\log n})$-round algorithm for graphs with arboricity up to $2^{O(\sqrt{\log n})}$.

For ruling sets, Kothapalli and Pemmaraju showed that one can compute a $2$-ruling set in $O(\log^{3/4} n)$ \cite{kothapalli-superfast2012} rounds on general graphs. Techniques from that paper, together with~\cite{GhaffariImproved16,BEPSv3}, yield the current state of the art of $O(\log^{1/4} \Delta+ O(\poly \log \log n)$, see \Cref{thm:related_work} for details on how to combine these works to obtain the claimed runtime. On trees it is possible to compute $2$-ruling sets in $O(\log \log n)$ rounds, and on graphs of girth at least $7$ in $\widetilde{O}(\log ^{5/3} \log n)$ rounds~\cite{BMU25}. 

For non-constant domination $\beta$, Gfeller and Vicari, combined with \cite{BEPSv3}, obtain an $O(\log \log n)$-ruling set algorithm in $O(\log \log n)$ rounds plus runtime of MIS on graphs with polylogarithmic degrees \cite{Gfeller07}. On graphs of girth at least $7$, it is possible to compute an $O(\log \log \log n)$-ruling set in $\widetilde{O}(\log \log n)$ rounds \cite{BMU25}.
The best known lower bound for randomized algorithms for general graphs is~$\Omega(\min \set{\beta \cdot \Delta^{1/\beta},\log_\Delta \log n})$~\cite{BBKO2021hideandseek}.

\subparagraph*{Deterministic \LOCAL.}The state of the art for MIS is the $\widetilde{O}(\log^{5/3} n)$ algorithm of Ghaffari and Grunau \cite{GG24}. In the same paper, they also present a $\widetilde{O}(\log n)$ round algorithm for $O(\log \log n)$-ruling sets. For graphs of arboricity $\alpha$  Barenboim and Elkin show that an MIS can be computed in $O(\alpha \cdot \sqrt{\log n}+ \alpha \cdot \log \alpha) $\cite{Barenboim2010}. Additionally, for graphs up to arboricity $O(\log^{1/2-\delta} n)$, $\delta >0$, they provide a sub-logarithmic algorithm to compute an MIS in $O(\log/\log \log n)$ rounds. For trees, their algorithm is tight due to a lower bound for MIS of $\Omega(\frac{\log n}{\log \log n})$ by \cite{BBKO2021hideandseek}. For $\beta$-ruling sets, a series of papers resulted in a lower bound of $\Omega(\min\{\beta \Delta^{1/\beta},\log_\Delta n\})$ which holds for general graphs~\cite{Balliu2019-focs-best,BBO22,BBKO2021hideandseek}. One of the earliest works on ruling sets presented an $O(\log n)$-round algorithm for $O(\log n)$-ruling sets~\cite{awerbuch89}. For constant $\beta$ the current state of the art algorithm runs in $O(\Delta^{2/(\beta+2)}+\log^*n)$~\cite{M21}.

\medskip
\noindent For ruling set algorithms in the \CONGEST model, we refer to \cite{MPU23}, and for the congested clique to \cite{CKPU23}. Next, we give an overview on prior ruling set algorithms in the context of low-space \mpc. For information related to linear-space \mpc, refer to \cite{JiKothappalliPemmarajuSingh25}.

\subparagraph*{Randomized low-space \mpc.} For general graphs, the current state of the art for MIS (1-ruling set) is by Ghaffari and Uitto \cite{GU19}, with a runtime of $O(\sqrt{\log \Delta} \cdot \log \log \Delta + \sqrt{\log \log n})$. For low arboricity graphs, the state of the art for MIS is by Ghaﬀari, Grunau, and Jin \cite{Ghaffari2020}, with a runtime of $O(\sqrt{\log \alpha} \cdot \log \log \alpha + \log \log n)$. For 2-ruling set, Kothapalli, Pai, and Pemmaraju \cite{Kothapalli2020} developed a $O(\log^{1/6} \Delta \cdot \log \log n)$-round algorithm using a technique that they call Sample-and-Gather. They also extend their result to $\beta$-ruling set with a runtime of $O(\beta \cdot \log^{1/(2^{\beta+1}-2)} \Delta \cdot \log \log n)$.

\subparagraph*{Deterministic low-space \mpc.} Most if not all deterministic algorithms are obtained by derandomizing prior randomized algorithms. Czumaj, Davies, and Parter \cite{CzumajDaviesParter20} provided an MIS algorithm with runtime $O(\log \Delta + \log \log n)$ via a graph sparsification technique that derandomizes Luby's algorithm. Fischer, Giliberti, and Grunau \cite{FischerGilibertiGrunau23} improved upon this for sparse graphs, giving a deterministic $O(\log \alpha + \log \log n)$-round algorithm for MIS and maximal matching for graphs with arboricity $\alpha$. Their key technique is a procedure that reduces the maximum degree to $\poly(\alpha)$ in $O(\log \log n)$ rounds.

Turning to ruling sets, Giliberti and Parsaeian \cite{GilibertiParsaeian24} gave the first deterministic ruling set algorithm with sublogarithmic round complexity, computing a $2$-ruling set in $O(\sqrt{\log \Delta} \cdot \log \log \Delta + \log \log^* n)$ rounds. This was later improved, for a certain range of $\Delta$, to $O(\log^{1/3} \Delta \cdot \log \log \Delta + \log \log n)$ rounds by Ji, Kothapalli, Pemmaraju, and Singh \cite{JiKothappalliPemmarajuSingh25}, who also extend their results to $\beta$-ruling set, with a runtime of $O(\log^{1/(2^{\beta}-1)} \Delta \cdot \log \log \Delta + \log \log n)$. By incorporating the degree drop technique by \cite{FischerGilibertiGrunau23}, they also formulate their result in terms of arboricity $\alpha$, resulting in a $O(\log^{1/(2^\beta - 1)} \alpha \cdot \log \log \alpha + \log \log n)$ runtime.


\subsection{Technical Overview and Challenges}
\label{subsec:tech_over}

The basic building block of our results is a variant of the classical MIS algorithm designed by Luby \cite{luby86} and independently by Alon, Babai, and Itai~\cite{alon86}.
\begin{tcolorbox}
  \emph{In each iteration, each (active) node $v$ draws a value $r_v \in (0,1)$ uniformly and independently at random and joins the independent set if and only if $v$ is a local minimum, i.e., $r_v < \min_{u\in N(v)}r_u$. Then, each node in the $2$-hop neighborhood is removed from the graph and considered as covered. }
\end{tcolorbox}
In the classical version for MIS, only nodes that are adjacent to the independent set nodes are removed from the graph, and it is shown that in each round, a constant fraction of edges is removed from the graph. We will call one such phase, of first drawing a random number, then letting local minima join, and finally removing the $2$-hop neighborhood, a \textit{Luby Phase}. Overall, this leads to an $O(\log n)$-round algorithm for the MIS, which is still the best known running time for general graphs to this day, for $\Delta=n^{O(1)}.$
  

In his seminal work~\cite{GhaffariImproved16}, Ghaffari presented a more involved algorithm that computes an MIS in $O(\log \Delta+\poly \log \log n)$ rounds, improving the state of the art for low degree graphs. Combined with the sampling techniques presented in~\cite{tushar-super-fast-2014}, this algorithm carries over to $2$-ruling sets and yields an $O(\sqrt{\log \Delta }+\poly \log \log n)$ algorithm (for general graphs). 

These results indicate that the main challenge in designing efficient algorithms arises primarily in the high-degree regime.

\medskip
Our main contribution is showing that, on low arboricity graphs, running repeated Luby Phases for $2$-ruling set yields a fast degree drop in the unsolved parts of the graph. 

\begin{restatable}[Degree Drop]{lemma}{degreeDrop}
    \label{lem:degreeDrop} Let $G$ be a graph with maximum degree $\Delta$ and arboricity $\alpha$. There is an $O(\log \log \Delta)$-round algorithm that computes an independent set $S$, such that $G\setminus (S\cup N_2(S))$ has maximum degree $\Delta' \leq  \max\{(82\alpha^4)^8,\log^8 n\}$~w.h.p.\firsttimefootnote{For better readability, we did not optimize constants. By repeatedly running our degree drop for more (but still constant) rounds, one can reduce the constants here.}
\end{restatable}

In fact, we show that after $O(1)$ slightly modified Luby Phases, the remaining graph has a maximum degree of $\Delta^{3/4}$ w.h.p.  Iterating this procedure yields the lemma above. 
Our procedure for the fast degree drop is not exactly the classical algorithm designed by Luby, but we adapt it by first letting only nodes of degree at most $\sqrt{\Delta}$ participate in a Luby Phase (Step 1) and then letting all nodes participate in another Luby Phase (Step 2). This process is then repeated for $O(1)$ rounds to guarantee a w.h.p.\ coverage of all high degree nodes. 
Although this two-step style is not a major change to Luby's algorithm, it leads to an exponentially better algorithm for the well-studied 2-ruling set problem. This highlights that there is no need for new, complex algorithms when classical ones are yet to be analyzed to their full potential.

As mentioned before, the main challenge is to deal with dependencies due to the existence of short cycles. If we were to ignore these dependencies, we see that we could achieve the desired fast degree drop by a straightforward analysis.

\subparagraph*{Analysis Baseline.}
We consider one Luby Phase, and our goal is to obtain a polynomial degree drop in each iteration. Fix a node $v$ with degree larger than $\Delta^{3/4}$. For simplicity of this overview, we also assume that the arboricity is constant.

First, consider the case where $v$ has few high-degree neighbors. In particular, $v$ has at least $\Delta^{3/4}/2$  neighbors with degree at most $\sqrt{\Delta}$. 
Any low-degree neighbor becomes a local minimum with probability roughly $1/\sqrt{\Delta}$. If the events of becoming local minima \emph{were independent}, at least one of the $\Delta^{3/4}/2$ low-degree neighbors would be a local minimum with high probability $1-(1-1/\sqrt{\Delta})^{\Delta^{3/4}/2}\geq 1-1/n^c$, using the assumption $\Delta\geq \poly\log n$.

In the second case, $v$ has many high-degree neighbors. Suppose that at least half of the neighbors have a degree of at least $\sqrt{\Delta}$. 
We show that in this case, due to small arboricity, node $v$ must have (roughly) $\Delta^{1.25}$ nodes in its 2-hop neighborhood.
If the events of becoming local minima \emph{were independent}, each one of those nodes would be a local minimum with probability roughly $1/\Delta$. 
Hence, again, at least one of those nodes would become a local minimum with high probability $1-(1-1/\Delta)^{\Delta^{1.25}}\geq 1-1/n^c$, and in both cases, node $v$ would be removed from the graph. This simplified view shows that the aforementioned two-step Luby style would not be required if the events were independent, but we will see that it is necessary for our analysis when the dependencies are present.

\subparagraph{Dealing with Dependencies via Concentration Inequalities.} Unfortunately, in both cases, the events of becoming local minima are anything but independent. 
In our approach, we deal with the dependencies created by short cycles directly.
Pemmaraju and Riaz~\cite{pemmaraju_et_al:LIPIcs.OPODIS.2016.9} were the first
to use read-$k$ concentration inequalities to analyze the progress of Luby's
algorithm on bounded-arboricity graphs.\footnote{We thank an anonymous reviewer for pointing us to the work of
Pemmaraju and Riaz~\cite{pemmaraju_et_al:LIPIcs.OPODIS.2016.9}. An earlier version of this manuscript handled the dependencies in the degree drop analysis by a substantially more involved direct argument. The read-$k$ tool allows us to present the algorithm and its analysis in a much cleaner form, while still obtaining near-optimal bounds for $2$-ruling sets on low arboricity graphs. Unfortunately, the framework presented in~\cite{pemmaraju_et_al:LIPIcs.OPODIS.2016.9} is missing reasoning in the proof of Theorem 5 why the underlying ranks stay independent after conditioning on nodes' ranks beating their out-neighbors' ranks. In particular, ranks can be correlated after this conditioning, see~\Cref{apx:counterexample}. We show how to resolve this issue in the proof of~\Cref{lem:node_joins}.}  

Their algorithm computes an MIS in
$O(\poly(\alpha) \cdot \sqrt{\log n \log \log n})$ rounds. While this is much faster
than the classical $O(\log n)$ bound in sparse graphs, it is still far from the
double-logarithmic runtimes suggested by ruling-set lower bounds. Moreover, the
polynomial dependence on $\alpha$ makes the result unsuitable for obtaining a
flat $O(\log \log n)$ runtime even for mildly growing arboricity.

Building on the read-$k$ concentration framework introduced by Pemmaraju and Riaz for analyzing Luby-style progress on bounded-arboricity graphs~\cite{pemmaraju_et_al:LIPIcs.OPODIS.2016.9}, we develop a degree drop analysis tailored to 2-ruling sets. 
To illustrate the structure of our argument, we first consider the simpler case in which a high-degree node has many low-degree neighbors. 
For the purpose of this overview, assume that the arboricity $\alpha$ is constant and recall that only nodes of degree at most $\sqrt{\Delta}$ are active in the Luby phase. 


\medskip
\textit{Our analysis.} Since the graph has arboricity $\alpha$, it admits an orientation in which every node has at most $\alpha$ outgoing edges. We fix such an orientation purely for the analysis. For every node $u$, we consider two events:
\begin{align*}
    O_u &: \text{\(u\) has smaller rank than all of its out-neighbors},\\[-5pt]
    I_u &: \text{\(u\) has smaller rank than all of its in-neighbors}.
\end{align*}

If both events occur, then $u$ has the smallest rank among all of its neighbors and therefore joins the ruling set. 

\begin{figure}
	\centering
	\includegraphics[width=0.85\textwidth]{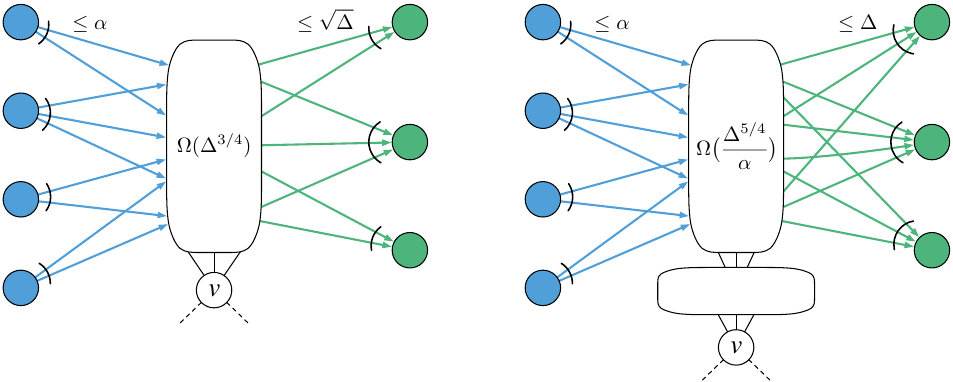}
    	\caption{On the left, we consider the case of many low-degree neighbors (discussed in the intro). On the right, we consider the case of many high-degree neighbors. Any blue node influences at most $\alpha+1$ events $I_u\coloneq \{ r(v) < \min \left( r(w) \mid w \textit{ in-neighbor of } v \right) \}$ for $u \in N(v)$, in both cases. Any green node influences at most $\sqrt{\Delta}+1$ (resp. $\Delta+1$) events $O_u\coloneq \{ r(v) < \min \left( r(w) \mid w \textit{ out-neighbor of } v \right) \}$ for $u \in N_2(v)$.}
	\label{fig:neighbors-both}
\end{figure}

\medskip

The analysis is split into two parts, first showing that enough nodes in the $2$-hop neighborhood of a large-degree node $v$ satisfy their event $O_u$, and then among those \emph{successful} in the first part at least one satisfies its respective event $I_u$, yielding that $v$ is covered.

\smallskip

\textit{First part of the analysis.} We first focus on the events $O_u$; see~\Cref{fig:neighbors-both} (left) for an illustration.  Since $u$ has at most $\alpha$ out-neighbors, the probability that it beats all of them is at least $1/(\alpha+1)$, which is $\Theta(1)$ when $\alpha$ is constant. Thus, among the many neighbors of $v$, we expect a constant fraction to satisfy $O_u$.

The main obstacle is that the events $O_u$ are not independent. However, the dependencies are limited as only nodes with degree at most $\sqrt{\Delta}$ are active and can only affect the events corresponding to itself and its neighbors. Thus, using a read-$k$ concentration bound with $k=\sqrt{\Delta}$, we show that the actual number of successful nodes remains close to its expectation. In particular, with high probability there are still $\Omega(\Delta^{3/4}/\alpha)$ successful nodes in $N(v)$.

\smallskip

\textit{Second part of the analysis.} Having identified many successful nodes, we now ask whether at least one of them also satisfies $I_u$. Since every active node has degree at most $\sqrt{\Delta}$, the probability that a successful node additionally beats all of its in-neighbors is at least $1/(\sqrt{\Delta}+1)$ and in expectation such a fraction joins the independent set.

Again, the events $I_u$ considered in the second part of the analysis are not independent, and additional dependencies on nodes' ranks are introduced due to conditioning on the nodes' success in the first part of the analysis; see~\Cref{apx:counterexample} for an example illustrating the new dependencies that have not been reasoned about in~\cite{pemmaraju_et_al:LIPIcs.OPODIS.2016.9}. We remove the latter dependencies by applying a version of the Harris-FKG inequality~\cite{EPW67,Zhao2022CorrelationInequalities}, effectively reducing to the case without any conditioning on success from the first part of the analysis. Then each (independent) rank affects only $\alpha+1$ of the events $I_u$; hence, among $\Omega(\Delta^{3/4}/\alpha)$ successful candidates that beat their out-neighbors, read-$k$ concentration implies that with high probability at least one also beats its in-neighbors and joins the independent set.


\smallskip

\textit{Case of many high-degree neighbors.} The case where $v$ has many high-degree neighbors can be handled in a very similar manner by considering nodes in $N_2(v)$, see~\Cref{fig:neighbors-both} (right) for an illustration. One can show that some node in this two-hop neighborhood joins the independent set with probability at least $1-1/n$. Repeating the process for $O(1)$ rounds then implies that every high-degree node is covered with probability at least $1-1/n^c$.

\subparagraph*{The Remaining Low-Degree Graph.}
\Cref{lem:degreeDrop} allows us to reduce the degree to $\poly \log n$, which allows us to use~\cite{GhaffariImproved16} to compute a partial independent set in $O(\log \Delta)=O(\log \log n)$ rounds such that the remaining unsolved graph shatters into components of size $O(\poly \log n)$.
Post-shattering based on~\cite{Barenboim2010} would compute an MIS for the remaining components in $O(\alpha^2+\log \log n)=O(\log\log n)$ rounds for instances up to arboricity $O(\sqrt{\log \log n})$. To prove~\Cref{thm:2RSconstArb}, we design a slightly different post-shattering by abusing the fact that we are computing a $2$-ruling set, making it work up to arboricity $O(\log \log n)$.

In fact, our degree drop lemma (\Cref{lem:degreeDrop}) applies w.h.p.\ to graphs of arboricity up to $\poly \log n$, reducing the maximum degree to $\poly \log n$. The main remaining challenge is handling this residual polylogarithmic-degree graph. This regime is difficult because both the maximum degree and the arboricity may still be $O(\poly \log n)$, which is already as hard as general graphs. Our \Cref{thm:MORE} addresses this case using the subsampling approach of~\cite{tushar-super-fast-2014}, which exploits the relaxed requirements of computing a $2$-ruling set.

\subparagraph*{Comparison to~\cite{BMU25}.} We already discussed the results of~\cite{BMU25}, but we want to emphasize that our algorithm recovers all of their results in a uniform way. For $\alpha=1$, we recover the result on trees by~\cite{BMU25}, additionally they present a $2$-ruling set algorithm for high-girth graphs with a runtime of $\widetilde{O}(\log^{5/3} \log n)$ rounds. 
Note that in graphs of girth at least $7$ the $3-$hop neighborhood of each node looks like a tree. In particular, if we root $N^3[v]$ at $v$, each node has only one in-neighbor. To deal with dependencies induced by in-neighbors, they implement an additional activation step; nodes are active with probability $1/2$, and only participate in the Luby Phase if they are active. Then, nodes with an inactive in-neighbor join if they beat all their out-neighbors, and these events are independent for such nodes due to the tree-like structure.
Although this is nice for a simpler analysis, it diverts more from the simple Luby like algorithm. Using read-$k$ analysis, we can recover their result without introducing any changes to our algorithm. In particular, our degree drop lemma also holds for graphs of girth at least $7$, and reduces the degree to $\poly \log n$ in $O(\log \log n)$ rounds. The higher runtime stems from the post-shattering, as described in~\cite{BMU25}.

\section{Preliminaries \& Notation} \label{sec:notation}

Given a graph $G=(V,E)$, we denote by $N(v)$ the neighborhood of $v$ (excluding $v$ itself). For a node $v\in V$ we let $\deg(v) \coloneqq |N(v)|$ denote the degree of $v$. Furthermore we denote by $N_2(v)$ all nodes $w \in V$ with $dist_G(v,w)\leq2$, i.e., the 2-hop neighborhood of $v$ (again excluding $v$ itself). By $N_2^-(v)$ we denote all nodes $w \in V$ with $dist_G(v,w)=2$, i.e., the \emph{exclusive} 2-hop neighborhood of $v$. Note that $N_2^-(v) \subset N_2(v)$ and $N_2^-(v) = N_2(v) \setminus N(v)$. For a subset $W \subset V$, let $G[W]$ be the induced subgraph of the subset $W$. For an oriented graph and a node $v$, we let $N^+(v)$ be the out-neighbors of $v$ and $N^-(v)$ be the in-neighbors of $v$. We denote $\{1,\dots,n\}$ by $[n]$ for $n\in \mathbb{N}$.

An algorithm is correct w.h.p. (with high probability) if there exists a constant $c>1$ such that it errs with probability $\leq 1/n^c$. We use the notation $\widetilde{O}$ as follows: $g(n) \in \widetilde{O}(f(n))$ if there exists a constant $c$ such that $g(n) \in O(f(n) \cdot \log^c(f(n)))$. Additionally, we will need some well-known results about graphs of bounded arboricity~\cite{Barenboim2010}, and~\Cref{lem:indep_set} is a standard result and we omit the proof.
  
\begin{lemma}\cite[Lemma 3.1]{Barenboim2010}
\label{lem:degeneratenodes}
    A graph $G=(V,E)$ with arboricity $\alpha$ has at least $\frac{|V|}{2}$ nodes with degree $4\alpha$ or less.
\end{lemma}

\begin{definition}
    \label{def:h-partition}
    Let $G=(V,E)$ be a graph with arboricity $\alpha$. A partition of the vertex set $H_1, \dots H_\ell$ is an $H-$partition with degree at most $(2+\epsilon)\cdot \alpha$ of size $\ell$ if every $v \in H_i$, $i \in \set{1,\dots,\ell}$, has at most  $(2+\epsilon)\cdot \alpha$ neighbors in the vertex set $\bigcup_{j=i}^{\ell}H_j$.
\end{definition}

\begin{lemma}\cite[Lemma 3.3 and 3.4]{Barenboim2010}
\label{lem:h-partition}
    For a graph $G$ with arboricity $\alpha$, and a parameter $\epsilon$, $0 < \epsilon \leq 2$, we can compute an $H$-partition of size $\ell \leq \lfloor \frac{2}{\epsilon} \log n \rfloor$ with degree at most $(2 + \epsilon) \cdot \alpha$ in $O(\log n)$ rounds.
\end{lemma}
\begin{lemma}
    \label{lem:indep_set}
    Let $G$ be a graph with arboricity $\alpha$ and let $M \subseteq V$ be a set of nodes. Then there exists an independent set $I\subseteq M$ with $|I| \geq |M|/2\alpha.$
    
\end{lemma}

\subparagraph*{Read-$k$ Families.}
We now define a read-$k$ family of random variables. Informally speaking, a read-$k$ family is a set of (possibly dependent) random variables $\{Y_j \mid 1 \leq j \leq m\}$ which are fully determined by an underlying set of independent random variables $\{X_i \mid 1\leq i \leq n\}$ such that each $X_i$ influences at most $k$ different random variables $Y_j$. For $k=1$ the $Y_j$'s are independent, and for $k>1$ the $Y_j$'s can have an arbitrary dependency structure. 

\begin{definition}[Read-$k$ families] Let $X_1, \dots,X_n$ be independent random variables. For $j \in [m]$, let $P_j \subseteq [n]$ and let $f_j$ be a Boolean function of $\left(X_i\right)_{i\in P_j}$. Assume that $|\{j\mid i\in P_j\}| \leq k$ for every $i \in [n]$. Then the random variables $Y_j=f_j\left((X_i)_{i \in P_j}) \right)$ are called a read-$k$ family.
\end{definition}

For read-$k$ families, \cite{https://doi.org/10.1002/rsa.20532} showed Chernoff-Like bounds allowing us to make statements about the concentration of such random variables.

\begin{theorem}\cite[Corollary of Theorem 1.1]{https://doi.org/10.1002/rsa.20532}
\label{lem:read-k-chernoff}
  Let $Y_1, \dots, Y_m$ be a family of read-$k$ indicator random variables with $\Pr[Y_i=1]=p_i.$ Let $Y \coloneq \sum_{i=1}^mY_i$. Then for any $\delta >0$,
  $$
  \Pr[Y\leq (1-\delta)\E[Y] ] \leq \exp\left(-\frac{\delta^2\E[Y]}{2k}\right).
  $$
\end{theorem}

\Cref{lem:read-k-chernoff} is not exactly stated as in \cite{https://doi.org/10.1002/rsa.20532}, but it can be derived from it. Similarly, \Cref{lem:read-k-all-p} can be derived from~\cite{https://doi.org/10.1002/rsa.20532}. For completeness, we include the derivations, see~\Cref{apx:deferred_proofs}.

\begin{theorem}\cite[Corollary of Theorem 1.2]{https://doi.org/10.1002/rsa.20532}
\label{lem:read-k-all-p}
    Let $Y_1, \dots, Y_m$ be a family of read-$k$ indicator random variables with $\Pr[Y_i=1]\leq p$. Then 
    $
    \Pr[Y_1=\dots=Y_m=1] \leq p^{m/k}.
    $
\end{theorem}

\section{Randomized 2-Ruling Set for Low Arboricity Graphs}

The goal of this section is to prove the following statement.

\thmRSloglogArb*

To prove this result, we first show how to reduce the degree of the graph to~$\Delta'=O(\poly \log n)$ in~$O(\log \log \Delta)$ rounds (see \Cref{sec:degree_drop}), by iteratively computing independent sets and removing their 2-hop neighborhoods from the graph. In \Cref{sec:shattering}, we handle the remaining low-degree graph using a known result of~\cite{GhaffariImproved16}, which leaves us with components of size~$N=O(\poly \log n)$. Finally, in \Cref{sec:small_comp}, we handle these small components by deterministically computing a 2-ruling set in each component in parallel, in~$O(\log N)=O(\log \log n)$ rounds.

\subsection{Stage I: Degree Drop}
\label{sec:degree_drop}

In this section, we show how to reduce the degree of the graph to~$\Delta'=O(\poly \log n)$ in~$O(\log \log \Delta)$ rounds by iteratively reducing the degree polynomially. In particular, we compute an independent set~$S$ in~$O(1)$ rounds such that all nodes of degree at least~$\Delta^{3/4}$ are contained in the 2-hop neighborhood~$S \cup N_2(S)$ w.h.p.\ Repeating this procedure for~$O(\log \log \Delta)$ rounds reduces the degree sufficiently. 
In particular, we will prove~\Cref{lem:degreeDrop} and restate it here for completeness.

{
\showrestatablefootnotesfalse
\degreeDrop*
}

In order to prove~\Cref{lem:degreeDrop}, we will show how we can reduce the max degree polynomially in $O(1)$ rounds. 
\begin{lemma}
\label{lem:degree_drop_iteration}
    Let~$G$ be a graph with arboricity~$\alpha$ and maximum degree~$\Delta \geq  \max\{(82\alpha^4)^8,\log^8 n\}$. There exists a procedure that computes an independent set~$S$ in~$O(1)$ rounds, such that~$G\setminus (S\cup N_2(S))$ has max degree~$\Delta'\leq \Delta^{3/4}$ w.h.p.\
\end{lemma}

The degree drop procedure for~\Cref{lem:degree_drop_iteration} is given by~\Cref{degree_drop_proc}. Running \Cref{degree_drop_proc} for $O(\log \log \Delta)$ rounds implies~\Cref{lem:degreeDrop}, as we will see at the end of this subsection.

For the rest of the section let~$U$ be the set of uncovered nodes after executing~\Cref{degree_drop_proc}, i.e.,~$U\coloneq \{ v \in V \mid v \in V \setminus (S \cup N_2(S)) \}$. 
In order to prove~\Cref{lem:degree_drop_iteration} we show that every node~$v$ with~$\deg(v)\geq \Delta^{3/4}$ is not in~$U$ w.h.p.\
\begin{algorithm}[!h]
   \caption{\label{degree_drop_proc}%
   Degree-Drop($G, \Delta,c$)}
    
 \newcommand\mycommfont[1]{\textcolor{gray}{\itshape #1}}
    \SetCommentSty{mycommfont}
    \SetKwComment{tcp}{}{}
    \DontPrintSemicolon

       ~$S \leftarrow \emptyset$

        \For{$i=0$ to~$c$}{
          \tcp*[l]{// Step 1}
       ~$A\leftarrow L = \{v \in V \mid \deg(v)\leq \sqrt{\Delta}\}$\tcp*[r]{// Active Nodes}
        
        Uniformly and independently at random compute a real~$r_v \in (0,1)$ for each~$v\in A$
        
         In parallel for all~$v\in A$
         
        \If{$r_v < r_w \forall w \in N(v) \cap A$}{
        
        ~$S \leftarrow S \cup \{v\}$
         
         Remove~$S \cup N_2(S)$ from~$G$
        }
        
        \tcp*[l]{// Step 2}
        
        Uniformly and independently at random compute a real~$r_v \in (0,1)
        $ for each~$v\in V$
        
        In parallel for all~$v\in V$
        
        \If{$r_v < r_w \forall w \in N(v)~$}{
        
       ~$S \leftarrow S \cup \{v\}$
        
         Remove~$S \cup N_2(S)$ from~$G$
        }
        }
\end{algorithm}
\subparagraph*{High Level Proof Idea.}
We will consider two different types of high degree nodes; either a node~$v$ with~$\deg(v)\geq \Delta^{3/4}$ has at least~$\deg(v)/2$ neighbors with degree at most~$\sqrt \Delta$ or at least~$\deg(v)/2$ neighbors with degree larger than~$\sqrt \Delta$. In both cases, we want to show that~$v$ is covered, after invoking \Cref{degree_drop_proc} with a constant $c$, with probability at least $1-\frac{1}{n^c}$.\ 

For the rest of the section, let 
\begin{align*}
    L = \{v \in V \mid \deg(v)\leq \sqrt \Delta \},  \quad&
    H_1= \{v\in V \mid \deg(v) \geq \Delta^{3/4}, |N(v) \cap L| \geq \deg(v)/2 \}, \\
    &H_2= \{v\in V \mid \deg(v) \geq \Delta^{3/4}, |N(v) \cap L| < \deg(v)/2 \},
\end{align*}
i.e.,~$H_1$ is the set of nodes that have degree at least~$\Delta^{3/4}$ and at least half of their neighbors have degree at most~$\sqrt \Delta$ and~$H_2$ is the set of nodes that have degree at least~$\Delta^{3/4}$ and less than half of their neighbors have degree larger than~$\sqrt \Delta$.
Observe that for every node~$v$ with~$\deg(v) \geq \Delta^{3/4}$ it holds that~$v\in H_1$ or~$v\in H_2$, so in particular it suffices to show that every node in~$H_1$ and~$H_2$ will be covered by~\Cref{degree_drop_proc} w.h.p.\

The proof works similarly for both cases; we consider a node $v$ 
and find a set $M(v)\subseteq N_2(v)$ of sufficient size w.r.t. the max degree of the active nodes in the corresponding step. The node $v$ then gets covered if at least one node $u \in M(v)$ joins the independent set, since the distance of $v$ to $u$ is at most $2$.
In particular, in Step 1 of~\Cref{degree_drop_proc}, only nodes with degree at most~$\sqrt{\Delta}$ are active, and we will find a set of size~$\Omega(\Delta^{3/4})$ in the neighborhood of~$v$. In Step 2 of~\Cref{degree_drop_proc} all nodes are active, so the max degree is~$\Delta$ and we will find a set of size~$\Omega(\Delta^{5/4}/\alpha)$ in the $2$-hop neighborhood of~$v$. In particular, both sets are polynomially (in $\Delta$) larger than the corresponding (active) max degree of the graph.  To show that at least one node of $M(v)$ joins $S$ with high probability, we use read-$k$ families.

\subparagraph*{Successful Nodes and Read-k Families.}

Throughout the analysis of both cases, we will fix an arbitrary $\alpha$-out-degree orientation. Notice that our algorithm does not rely on or use this orientation, but it is solely used for the analysis. With respect to this fixed orientation, we can define the notation of successful nodes. Recall that for an oriented graph and a node $v$, we denote $N^+(v)$ as the out-neighbors of $v$ and $N^-(v)$ as the in-neighbors of $v$.

\begin{definition}
    Let $G=(V,E)$ be an oriented graph. We say a node $v$ is successful in one iteration of Step 1 or 2 if its real $r(v)$ is smaller than the real of all its out-neighbors, i.e., $r(v)<\min\{r(u)\mid u \in N^+(v)\}$. We say a set of nodes is successful if all nodes of that set are successful.
\end{definition}
In both cases, we will consider a fixed set $M(v) \subseteq N_2(v)$ and want to argue that at least one node joins the independent set $S$ w.h.p. First, we will show that a large fraction of the set $M(v)$ will be successful nodes. Second, we will show that one of those nodes also joins with a large probability. We formulate both of these results in a more general setting, and will see in~\Cref{sec:many_low_degree} and~\Cref{sec:many_high_degree} how to apply the statements to each of the cases.

\begin{lemma}
\label{lem:many_successful}
    Let $G=(V,E)$ be a graph with arboricity $\alpha$, and (active) max degree $\Lambda$. Let $M\subseteq V$ be a fixed subset of (active) nodes. Then, with probability at least
    $$1- \exp\left(-\frac{|M|}{8(\alpha+1)(\Lambda+1)}\right) $$
    at least $|M| /2 (\alpha+1)$ nodes of $M$ are successful.
\end{lemma}
\begin{proof}
    Recall that we have a fixed $\alpha$-outdegree orientation for the analysis. We define the events $E_v=\{r(v) < \min\{r(u) \mid u \in N^+(v)\} \}$ for all $v \in M$, and let $Y_v$ be the indicator r.v. of $E_v$. Note that every event $E_v$ is fully determined by the ranks $r(u)$ of the nodes $u\in \left(v \cup N^+(v) \right)$, i.e., $Y_v=f_v(r(u)_{u\in \left(v \cup N^+(v) \right)})$ for some boolean functions $f_v$. Since all active nodes have degree at most $\Lambda$, each rank $r(v)$ influences at most $\Lambda+1$ many r.v. $Y_v$; the r.v. $Y_u$ for $u\in N^-(v)$ and $Y_v$ if $v \in M$. Therefore, $\{Y_v \mid v \in M\}$ and $\{r(v) \mid v \textit{ active}\}$ form a read-$(\Lambda+1)$ family.
    Let $Y\coloneq \sum_{v \in M} Y_v $ be the number of successful nodes in $M$ and observe that 
    $ \P[Y_v=1] \geq \frac{1}{\alpha+1}$, 
    since each node has at most $\alpha$ out neighbors. Therefore, $\E[Y]\geq \frac{|M|}{\alpha+1}$. Thus, we can apply~\Cref{lem:read-k-chernoff} to obtain
    \begin{equation*}
        \P[Y\leq |M|/2(\alpha+1)]\leq \exp\left(-\frac{\left(\frac{1}{2}\right)^2|M|/(\alpha+1)}{2(\Lambda+1)}\right) \leq \exp\left(-\frac{|M|}{8(\alpha+1)(\Lambda+1)}\right),
    \end{equation*} 
    and therefore $\P[Y \geq |M|/2(\alpha+1)] \geq 1- \exp\left(-\frac{|M|}{8(\alpha+1)(\Lambda+1)}\right)$.
\end{proof}

In both cases we will use that result to show that we can obtain a large number of successful nodes in the $2$-hop neighborhood of any node $v$ with degree $\geq \Delta^{3/4}$. In order to show that $v$ gets covered, it suffices to show that one of those successful nodes joins the independent set. For that, at least one node also has to have a smaller rank $r(u)$ than all of its in-neighbors. Thus, you may think of $M$ in the statement as the successful $2$-hop neighbors of a node we want to cover. 

\begin{lemma}
\label{lem:node_joins}
     Let $G=(V,E)$ be a graph with arboricity $\alpha$, and (active) max degree $\Lambda$. Let $M\subseteq V$ be a fixed subset of successful, independent (active) nodes. Then, with probability at most 
     $$
         \left(1-\frac{1}{\Lambda+1}\right)^{|M|/\alpha}
     $$
    no node in $M$ chooses a rank smaller than all its in-neighbors.
\end{lemma}
\begin{proof}
    Again, recall that we have a fixed $\alpha$-outdegree orientation for the analysis. We define the events $F_v = \{ r(v)<\min\{r(u) \mid u \in N^-(v) \}\}$ and $E_v=\{r(v) < \min\{r(u) \mid u \in N^+(v)\} \}$ for all $v \in M$, and let $Z_v$ be the indicator r.v. of $F_v^c$.

    Further define $F=\bigcup_{v \in M} F_v$ and $E=\bigcap_{v\in M} E_v$, i.e., $E$ is the event that $M$ is successful and $F$ is the event that at least one node of $M$ also beats its in-neighbors. First, we aim to upper bound $\P[F^c\mid E]$ by the unconditional probability $\P[F^c]$. To this end define the transformed r.v.
    $$
    X_v = \begin{cases}
        1-r(v), & \text{if } v \in M,\\
        r(v),   & \text{if } v \notin M.
    \end{cases}
    $$
    Since the ranks $r(v)$ are mutually independent, the r.v. $X_v$ are mutually independent as well. Further, since $M$ is an independent set, every neighbor $z \in N(M)$ satisfies $z \notin M$. Thus, for every $v \in M$ and every $z \in N(v)$, we have
    $
    r(v)<r(z) \Leftrightarrow X_v+X_z>1.
    $    Consider the coordinate-wise partial ordering of $X=(X_v)_{v\in V}$: We say $X \leq X'$ if $X_v \leq X'_v$ coordinate-wise (for all $v\in V$).
    Observe that the indicator r.v. of the event $\{X_u+X_z>1\}$ is non-decreasing in $X$ w.r.t. the coordinate-wise partial ordering of $X$. Every event $F_v$ and $E_v$ is an intersection of such events and thus their indicator r.v.'s are non-decreasing. Similarly, the indicator r.v. $\mathds{1}_E$ of $E$ and $\mathds{1}_F$ of $F$ are non-decreasing in $X$, as they are an intersection of $E_u$'s, or a union of $F_u$'s, respectively.

       Since the variables \((X_v)_{v\in V}\) are mutually independent, we can apply~\Cref{thm:associated} and obtain that $Cov(h(X),g(X) )=\E[h(X)g(X)]-\E[h(X)]\E[g(X)]\geq 0$  for any pair of non-decreasing functions $h,g$. Therefore, by setting $h=\mathds{1}_E$ and $g=\mathds{1}_F$ we obtain $E[\mathds{1}_E\mathds{1}_F]\geq \E[\mathds{1}_E]\E[\mathds{1}_F]$. By using the fact that  $\E[\mathds{1}_E\mathds{1}_F]=\E[\mathds{1}_{E\cap F}]$, this yields $\P[E\cap F]  \geq \P[E]\Pr[F]$. Since $\P[E] >0$, we get $\P[F\mid E]  \ge \P[F] $, or equivalently $  \P[F^c \mid E]  \le \P[F^c].$

\smallskip
    
    As the second step we upper bound the unconditional probability $\P[F^c]$ using read-$k$ families. Note that every event $F_v^c$ is fully determined by the independent (we do not condition on $M$ being successful in this step, but consider the unconditional probability) ranks $r(u)$ of the nodes $u\in \left(v \cup N^-(v) \right)$, i.e., $Z_v=f_v(r(u)_{u\in \left(v \cup N^-(v) \right)})$ for some boolean functions $f_v$. Note that the in-neighbors of a node $v\in M$ are not in $M$. Since every node has at most $\alpha$ out-neighbors, each (active) node $v$ influences at most $\alpha$ different r.v. $Z_v$; the r.v. $Z_u$ for $u \in N^+(v)$ if $v \notin M$ and only $Z_v$ if $v \in M$. Therefore, $\{Z_v \mid v \in M\}$ and $\{r(v) \mid v \textit{ active}\}$ form a read-$\alpha$ family. Here it is crucial that the underlying ranks are independent, which would not be the case if we were to consider the ranks after conditioning on $M$'s success, see~\Cref{apx:counterexample}. Furthermore, observe that 
    $        \P[Z_v=1] = 1 -\frac{1}{\deg^-(v)+1} \leq 1-\frac{1}{\Lambda+1}$.
    
    Thus, we can apply \Cref{lem:read-k-all-p} and bound the probability that no node $v \in M$ chooses a rank smaller than all its in-neighbors by
    \begin{equation*}
        \P[F^c\mid E] \leq \P[F^c] = \P[\bigwedge_{v \in M}Z_v =1] \leq \left(1-\frac{1}{\Lambda+1}\right)^{|M|/\alpha}. \qedhere
    \end{equation*}
\end{proof}

We have set up all the general statements that hold for any graph of arboricity~$\alpha$ and can show for the first case~$v\in H_1$ that~$v$ gets covered w.h.p.\

\subsubsection{Many Low-Degree Neighbors}
\label{sec:many_low_degree}
In this subsection, we show that a node~$v\in H_1$ stays uncovered with probability at most $1/n$ during one iteration of Step 1 of~\Cref{degree_drop_proc}. So for the rest of this subsection, we assume that only nodes in~$L$ are active; in particular, we know that the max degree of the active nodes is~$\sqrt{\Delta}$. 

\begin{lemma}[Many Low-Degree Neighbors]
\label{lem:many_low_deg_neighbors}
Let $\Delta \geq \max\{(66\alpha^3)^{8}, \log^8 n\}$ and $U^{it}$ be the set of uncovered nodes after one arbitrary iteration of the for-loop of~\Cref{degree_drop_proc}. For a node~$v \in H_1$ at the beginning of the iteration, we have~$\P[v \in U^{it}] \leq \frac{1}{n}$.
\end{lemma}

\begin{proof}
    Let~$L(v)\coloneq N(v) \cap L$ and observe that~$|L(v)| \geq \frac{\Delta^{3/4}}{2}$, since~$v \in H_1$.
    By~\Cref{lem:indep_set} there exists an independent set $I \subseteq L(v)$ with $|I|\geq |L(v)|/2\alpha \geq \Delta ^{3/4}/4\alpha$.

    First, we will find a large subset of independent successful nodes in the neighborhood of $v$.
    Recall that in Step 1 of \Cref{degree_drop_proc} only nodes with degree less than $\sqrt{\Delta}$ are active. Now we apply \Cref{lem:many_successful} to $I$ and obtain a set of successful nodes $M(v)$ of size $|M(v)| \geq |I|/2(\alpha+1)\geq\Delta^{3/4}/9\alpha^2$ with probability at least 
    \begin{align*}
        1-\exp\left(-\frac{|I|}{8(\alpha+1)(\sqrt{\Delta}+1)}\right) & \geq 1-\exp\left( -\frac{\Delta^{3/4}}{64\alpha^2(\alpha+1)(\sqrt{\Delta}+1)} \right)\\&  = 1-\exp\left(-\frac{\Delta^{1/4}}{66\alpha^3}\right)  \geq 1- \exp\left( - \log n\right)=1- \frac{1}{n}.
    \end{align*}
    Second, we upper bound the probability that none of those successful nodes joins the independent set.
    By~\Cref{lem:node_joins} we obtain that no node of $I$ joins the independent set with probability at most 
    \begin{align*}
      \left(1-\frac{1}{\sqrt{\Delta}+1}\right)^{|M(v)|/\alpha} \leq \exp\left(-\frac{|M(v)|}{\alpha(\sqrt{\Delta}+1)}\right)& \leq \exp\left(-\frac{\Delta^{3/4}}{10\alpha^3\sqrt{\Delta}} \right) \\& \leq \exp\left(-\frac{\Delta^{1/4}}{10\alpha^3} \right) \leq \exp(-\log n).
    \end{align*}
   Thus $\P[v \in U^{it} ]\leq \P[\textit{no node of } M(v) \textit{ joins the independent set}] \leq \frac{1}{n}.$
\end{proof}

\subsubsection{Many High-Degree Neighbors}
\label{sec:many_high_degree}
In this subsection, we show that a node~$v\in H_2$ stays uncovered with probability at most $1/n$ during one iteration of Step 2 of~\Cref{degree_drop_proc}. In contrast to the prior case, all the nodes are active and thus the max degree is~$\Delta$. Formally, we show the following. 

\begin{lemma}[Many High-Degree Neighbors] \label{lem:many_high_deg_neighbors}
    Let $\Delta \geq \max\{(82\alpha^4)^8, \log^8 n\}$ and $U^{it}$ be the set of uncovered nodes after one arbitrary iteration of the for-loop of~\Cref{degree_drop_proc}. For a node~$v \in H_2$ at the beginning of the iteration, we have~$\P[v \in U^{it}] \leq \frac{1}{n}$.
\end{lemma}

The proof for this case is very similar to the proof we have just seen, but we need to use arboricity in one more step to obtain a set of sufficient size within the 2-hop neighborhood of~$v$. 

\begin{lemma}[Large 2-Hop Neighborhood] \label{lem:large_2_hop_NH}
    Let~$G$ be a graph with arboricity~$\alpha$ and max degree~$\Delta\geq 5\alpha^2$.
    Consider a node~$v \in H_2$. Then~$|N_2^-(v)| \geq \frac{\Delta^{5/4}}{5 \alpha}$.
\end{lemma}

\begin{proof}
     Consider the induced subgraph~$G[N(v)\cup N_2^-(v)]$ and let $m_N$ denote the number of edges on this subgraph. Then we obtain that
    \begin{align*}
         m_N \leq \alpha(\deg(v)+|N_2^-(v)|) 
        \iff  |N_2^-(v)| \geq \frac{m_N}{\alpha} - \deg(v) .
    \end{align*}
    We know that~$m_N \geq \deg(v)\cdot \sqrt\Delta/4$ and~$\deg(v) \geq \Delta^{3/4}$, which implies that 
    \begin{align*}
        |N_2^-(v)|\ge \Delta^{5/4}/4\alpha - \Delta^{3/4}.
    \end{align*}
    In particular, for~$\Delta \geq 5\alpha^2$, we have that~$|N_2^-(v)|\geq \Delta^{5/4}/5\alpha$.
\end{proof}



\begin{proof}[Proof of~\Cref{lem:many_high_deg_neighbors}]

By~\Cref{lem:large_2_hop_NH}, we know that $|N_2^-(v)|\geq \Delta^{5/4}/5\alpha,$ since $v \in H_2$.  Thus, by~\Cref{lem:indep_set} there exists an independent set $I \subseteq N_2^-(v)$ with $|I|\geq |N_2^-(v)|/2\alpha \geq \Delta ^{5/4}/10\alpha^2$. 

Now we can apply \Cref{lem:many_successful} to $I$ and obtain a set of successful nodes $H(v)$ of size $|H(v)|\geq|I|/2(\alpha+1)\geq\Delta^{5/4}/21\alpha^3$ with probability at least 
    \begin{align*}
        1-\exp\left(-\frac{|I|}{8(\alpha+1)(\Delta+1)}\right) & \geq 1- \exp\left( -\frac{\Delta^{5/4}}{80\alpha^2(\alpha+1){(\Delta+1)}} \right) \\&= 1-\exp\left(-\frac{\Delta^{1/4}}{82 \alpha^3}\right)  \geq1- \exp\left( - \log n\right)= 1-\frac{1}{n}. 
    \end{align*} 

 Second, we upper bound the probability that none of those successful nodes joins the independent set.
By~\Cref{lem:node_joins} we obtain that no node of $H(v)$ joins the independent set with probability at most
    \begin{align*}
        \left(1-\frac{1}{\Delta+1}\right)^{|H(v)|/\alpha} \leq \exp\left(-\frac{|H(v)|}{\alpha(\Delta+1)}\right) \leq \exp\left(-\frac{\Delta^{5/4}}{22\alpha^4{\Delta}} \right) & \leq \exp\left(-\frac{\Delta^{1/4}}{22\alpha^4} \right) \\ & \leq \exp(-\log n).
    \end{align*}
   Thus $\P[v \in U^{it} ]\leq \P[\textit{no node of } H(v) \textit{ joins the independent set}]\leq \frac{1}{n}.$
\end{proof}

Now we are ready to prove \Cref{lem:degree_drop_iteration}

\begin{proof}[Proof of \Cref{lem:degree_drop_iteration}]
    Let $c$ be the constant~\Cref{degree_drop_proc} is called with.
    
    Since degrees only reduce during~\Cref{degree_drop_proc}, a node $v$ starting with $\deg(v) > \Delta^{3/4}$ at the beginning of the call of~\Cref{degree_drop_proc} might have a degree $\leq \Delta^{3/4}$ after some iterations of the for-loop. Then $v$ did not get covered, but its degree still dropped far enough and is at most $\Delta^{3/4}$.

    Thus, we may assume the degree of a node $v$ with $\deg(v)>\Delta^{3/4}$ does not drop below $\Delta^{3/4}$ during the call of~\Cref{degree_drop_proc}. As long as $\Delta\geq \max\{(82\alpha^4)^8, \log^{8}n\}$ \Cref{lem:many_low_deg_neighbors} and~\Cref{lem:many_high_deg_neighbors} hold. Then for $v$, since $\deg(v)> \Delta^{3/4}$ at the beginning of each iteration of the for-loop, it either holds that $v \in H_1$ or $v \in H_2$. Since we always sample new random values at the beginning of each iteration, we can either apply~\Cref{lem:many_low_deg_neighbors} or~\Cref{lem:many_high_deg_neighbors} to each $v$ in each iteration. Thus 
    \begin{align*}
        \P[v\in U] \leq (\P[v \in U^{it}])^c \leq \frac{1}{n^c}.
    \end{align*}
    Therefore, $v$ gets covered w.h.p.\ and removed from the graph after~\Cref{degree_drop_proc} and the result follows by a union bound over all high degree nodes.
\end{proof}

    Finally, we can prove the Degree Drop Lemma, which we restate here for readability.

    {
\showrestatablefootnotesfalse
\degreeDrop*
}
    
    \begin{proof}[Proof of~\Cref{lem:degreeDrop}]
    We run~\Cref{degree_drop_proc} for~$T\coloneq O(\log \log \Delta)$ rounds, always calling it on the remaining graph after the previous iteration. Let $S_1,\dots,S_T$ be the ruling sets computed in these calls and let $S\coloneq \left(\bigcup_{i=1}^T S_i \right)$. This algorithm runs in $O(\log \log \Delta)$ rounds, since each call takes $O(1)$ rounds. The max degree of the remaining graph is $\max\{\Delta^{(3/4)^{T}},(82\alpha^4)^{8},\log^8n\}$ w.h.p.\ by combining~\Cref{lem:degree_drop_iteration} with a union bound over the $T$ iterations. Since $\Delta^{(3/4)^{T}}=O(1)$, we have that the max degree of the remaining graph is $\max\{(82\alpha^4)^{8},\log^8n\}$ w.h.p.\
\end{proof}

When the arboricity of the given graph is bounded by $O(\poly \log n)$, we get a degree reduction to $O(\poly \log n)$.

\begin{corollary} \label{cor:degree_drop}
     Let~$G=(V,E)$ be a graph of arboricity~$O(\poly \log n)$ and with max degree~$\Delta$. There is an~$O(\log \log \Delta)$ round algorithm that takes~$G$ as input and returns an independent set~$S$ such that~$G'=G\setminus(S \cup N_2(S))$ has max degree~$\Delta'=O(\poly \log n)$ w.h.p.\ 
\end{corollary}

\begin{proof}
    This follows by~\Cref{lem:degreeDrop} by observing that $\max\{(82\alpha^4)^8, \log^8n\}=O(\poly \log n)$ when $\alpha=O(\poly \log n)$.
\end{proof}

\subsection{Stage II: Shattering} \label{sec:shattering}

After applying~\Cref{cor:degree_drop}, we can assume that we are given a graph of max degree~$\Delta=O(\poly \log n)$. In this case, we can use the existing work of \cite{GhaffariImproved16} to compute an MIS such that the remaining graph consists only of small components.

\begin{lemma}[Low max degree]
    \label{lem:low-max-degree}
    Let~$G$ be a graph with max degree~$\Delta \leq \log^{c} n$, for some constant $c>0$. There is a~$O(\log \log n)$ round algorithm that computes an independent set~$S$ such that all connected components of~$G\setminus (S\cup N(S))$ have size~$O(\poly \log n)$ w.h.p.\
\end{lemma}

\begin{proof}
    Run Ghaffari MIS algorithm \cite{GhaffariImproved16}, which computes an MIS, such that after~$O(\log \Delta)=O(\log \log n)$ rounds the remaining graph consists of components of size~$N \coloneqq O(\log_{\Delta} n \cdot \Delta^4)=O(\poly \log n)$ w.h.p.\
\end{proof}

\subsection{Stage III: Deterministic 2-Ruling Set for Small Components}
\label{sec:small_comp}
After applying~\Cref{cor:degree_drop} and \Cref{lem:low-max-degree}, we are left with a graph that consists of components of size~$N \coloneqq O(\poly \log n)$. In this subsection, we present a simple and self-contained algorithm to compute a 2-ruling set that we can apply to compute a~$2$-ruling set on all the connected components in parallel in~$O(\log N)=O(\log \log n)$ rounds, since we assume the arboricity to be~$O(\log \log n)$.

\begin{theorem}
\label{thm:det2RS}
    Given a graph with arboricity~$\alpha$, there exists a deterministic \CONGEST algorithm computing a 2-ruling set in~$O(\log n+\alpha)$ rounds.
\end{theorem}

Even though the runtime can also be achieved by combining prior work, e.g. by computing a~$O(\alpha^2)$-coloring using \cite{Barenboim2010} and then choosing~$B=\alpha=O(\log n)$ in \cite{KMW18-rulingSets} to compute a~$2$-ruling set in~$O(B\log_BC)=O\left(\alpha \frac{2 \log \alpha}{\log \alpha}\right)=O(\log n)$ rounds, we think that the algorithm may be of independent interest.

The idea of the algorithm is to use the well-known notion of~$H$-partition, \Cref{def:h-partition}, and to heavily abuse the fact that we are computing a~$2-$ruling set. In particular, we will deactivate nodes such that each of the inactive nodes is adjacent to an active node and the induced graph on the active nodes consists of connected components with degree at most~$4\alpha$. Which allows us to compute an MIS on each connected component efficiently in parallel. 

\begin{algorithm}
    \caption{Deterministic 2-RS}
    \label{alg:det-2-RS-low-arb}
    \SetKwInOut{Input}{Input}
    \SetKwInOut{Output}{Output}
    
    \newcommand\mycommfont[1]{\textcolor{gray}{\itshape #1}}
    \SetCommentSty{mycommfont}
    \SetKwComment{tcp}{}{}
    \DontPrintSemicolon

        Compute H-partition with degree~$4\alpha$ of size~$\ell \leq \lfloor \log n \rfloor$ 
       ~$A \leftarrow V$ \tcp{---Active nodes}
        
        \For{$i=\ell$ down to~$1$} {
        In parallel for each node ~$v\in V(H_i) \cap A$, remove all nodes in~$N(v)\cap V(H_j)$ with~$j<i$ from~$A$
        }
         Compute~$\Delta+1 = O(a)=O(\log n)$ coloring on~$G[A]$ 
         
         Compute MIS by standard reduction to coloring 
    
\end{algorithm}

\begin{proof}[Proof of \Cref{thm:det2RS}]
    The pseudocode for the algorithm is given by \Cref{alg:det-2-RS-low-arb}.
    We start by computing a~$H$-partition~$H_1,\dots, H_\ell$ with degree at most~$4\alpha$ of size~$\ell\leq \lfloor \log n \rfloor = O(\log n)$ in~$O(\log n)$ rounds, see~\Cref{lem:h-partition}. 
    Let~$A=V$ denote the active nodes. 
    Iterate through the layers from~$\ell$ down to~$1$. In iteration~$i$ each active node~$v \in A \cap V(H_i)$ deactivates all its neighbors that are in earlier layers, i.e., we remove~$u \in A \cap V(H_j)$ with~$j<i$ from~$A$. This step can be performed in parallel for all nodes in~$V(H_i)$ and thus only takes~$O(1)$ rounds; therefore, iterating through all layers takes~$O(\ell)=O(\log n)$ rounds. 
    \begin{claim*}
       ~$G[A]$ consists of connected components, each with maximum degree at most~$4\alpha$ and each node~$V\setminus A$ is adjacent to a node in~$A.$
    \end{claim*}  
    If the claim holds, we can compute a~$\Delta +1=O(\alpha)$ coloring on~$G[A]$ in~$O(\Delta + \log^* n)=O(\alpha+\log^* n)$ rounds \cite{10.1145/3486625} and turn it into an MIS in~$O(\Delta+1)=O(\alpha)$ rounds by iterating through the color-classes. The MIS~$S$ of~$G[A]$ is a~$2$-ruling set of~$G$, since every node is adjacent to a node in~$A$ and therefore $dist(v,S)\leq2$ for all nodes~$v\in V(G)$. The overall runtime then is~$O(\log n + \alpha)$ rounds, by first computing the~$H$-partition and then the MIS on~$G[A]$.

    In the remainder of this proof, we show that the claim holds. First, consider any node~$v$ that is active at the end of the algorithm. Assume~$v \in V(H_i)$ and observe that~$v$ was active in iteration~$\ell -i+1$, i.e., when we considered layer~$H_i$. Thus~$v$ deactivated all its neighbors from earlier layers, which implies~$\deg_A(v) \leq |N(v) \cap (\bigcup_{j=i}^{\ell}H_j)| \leq 4\alpha$ since the~$H-$partition has degree~$4\alpha$.
    Secondly, consider any node~$v$ that got deactivated in the course of the algorithm, let~$i$ be the iteration in which~$v$ got deactivated. Let~$d(v)$ be the node that deactivated~$v$. Since~$d(v)$ deactivated~$v$ in iteration~$i$ we know that~$d(v)\in H_{\ell-i+1}$. Furthermore, we know that only nodes in layers~$H_j$ with~$j < \ell-i+1$ get deactivated throughout the rest of the algorithm. So~$d(v)$ is an active neighbor of $v$ at the end of the algorithm.
 \end{proof}

We can use this deterministic~$2-$ruling set algorithm for our remaining instances of size~$N=\poly \log n$ and arboricity~$O(\log \log n)$ in a straightforward way to obtain the following.

\begin{corollary}
\label{cor:det2RS}
    Given a graph of size~$N=O(\poly \log n)$ with arboricity~$O(\log \log n)$, there exists a deterministic \CONGEST algorithm computing a 2-ruling set in~$O(\log \log n)$ rounds.
\end{corollary}

Finally, we can prove \Cref{thm:2RSconstArb}.

\begin{proof}[Proof of \Cref{thm:2RSconstArb}]
    Consider the following algorithm:
    We run the procedure of~\Cref{cor:degree_drop} computing an independent set $S$. Second, we invoke the procedure of \Cref{lem:low-max-degree} on $G'\coloneq G\setminus (S\cup N_2(S))$ to compute an independent set $I$. Finally, we invoke the procedure of~\Cref{cor:det2RS} on all the components of $G'\setminus (I \cup N(I))$ in parallel to compute a ruling set on each individual one. Our final $2$-ruling set is the union of $S, I$ and the ruling set of each of the components.
    
    \textbf{Correctness of the Subroutines.} After the call of the first procedure, we are left with a graph of max degree $O(\poly \log n)$  w.h.p.\, by~\Cref{cor:degree_drop}. Thus, we can invoke the procedure of~\Cref{lem:low-max-degree} which leaves us with small components of size $O(\poly \log n)$ w.h.p.\ This in turn lets us call the procedure of~\Cref{cor:det2RS} on each component in parallel.
     
    \textbf{Runtime.} By~\Cref{cor:degree_drop} after $O(\log \log \Delta)$ rounds $G'$ has a max degree of $\poly \log n$ w.h.p.\ Therefore, the call of the procedure of~\Cref{lem:low-max-degree} takes $O(\log \log n)$ rounds and all connected components of $G'\setminus(I \cup N(I))$ have size $O(\poly \log n)$ w.h.p.\ Thus, and since the procedure of~\Cref{cor:det2RS} runs in parallel on each connected components, the computation of the independent sets on each component runs in $O(\log \log n)$ rounds.  
   
    \textbf{Independence and Domination.} In each called procedure, we always compute independent sets and remove each computed independent set and its $2$-hop neighborhood before computing the next independent set, ensuring independence across the different calls of procedures. 
    Before the final computation of ruling sets for each component, we removed nodes dominated by distance at most $2$, so we computed a $2$-ruling set in the end.
    
    \textbf{Small messages.}
    Finally, observe that every subroutine for Stage II and III works in the \CONGEST model. For Stage I, we can use fractional values instead of reals for the ranks using $c \cdot \log n=O(\log n)$ bits. Then we have no collisions w.h.p. $1-1/n^{c-2}$ by a union bound. This completes the proof of~\Cref{thm:2RSconstArb}.
\end{proof}

\begin{remark*}
    The reason why we require arboricity of~$O(\log \log n)$ for~\Cref{thm:2RSconstArb} is that it allows us to solve the shattered components of small size efficiently. If we allowed arboricity of~$O(\poly \log n)$, we would end up with instances of size~$O(\poly \log n)$ and arboricity also~$O(\poly \log n)$, which correspond to general graphs. Thus, solving these components in $O(\log \log n)=O(\log N)$ rounds seems to be no easier than designing an $O(\log n)$ deterministic algorithm for $2$-ruling sets on general graphs.
\end{remark*}
\section{Randomized 2-Ruling Set for Graphs of General Arboricity}
\label{sec:lognArb}

In this section, we are considering graphs of general arboricity~$\alpha \gg  \log \log n$. In particular, we will show the following result in this section.

\thmMORE*

In~\Cref{cor:degree_drop} we have already established that we can reduce the max degree to $O(\poly \log n)$ when arboricity is $O(\poly \log n)$. Therefore, we will consider the case where~$ \alpha = \Omega(\poly \log n)$. In this situation, we derive a degree drop down to $\poly \alpha$ from~\Cref{lem:degreeDrop}.

\begin{corollary} \label{cor:degreeDropArb}
    Let~$G$ be a graph with max degree~$\Delta$ and arboricity~$\alpha=\Omega(\poly \log n)$. There is an~$O(\log \log\Delta)$-round procedure that computes an independent set~$S$ such that~$G\setminus (S\cup N_2(S))$ has max degree~$\Delta'\leq  82^8\alpha^{32}=O(\poly (\alpha))$ w.h.p.\
\end{corollary}

\begin{proof}
    The statement follows immediately by~\Cref{lem:degreeDrop}, since $\max\{(82\alpha^4)^8, \log^8n\}=82^8\alpha^{32}$, already for $\alpha\geq \log n$.
\end{proof}

The procedure is also given by running~\Cref{degree_drop_proc} for $O(\log \log \Delta)$ rounds. Finally, we will use sampling techniques by \cite{tushar-super-fast-2014} to reduce the degree further and finally show how to compute a $2$-ruling set.

\begin{lemma}[Theorem 1~\cite{tushar-super-fast-2014}]
\label{lem:spars}
    Let $ G$ be an arbitrary $n$-vertex graph with maximum degree $\Delta$ and $f$ be a parameter. With high probability, we can compute a vertex-subset $S \subseteq V (G)$ in $O(\log_f \Delta)$ rounds such that $\Delta(G[S])=O(f\cdot \log n)$, and every vertex in $ V$ is either in $S$ or has a neighbor in $S$.
\end{lemma}

\begin{lemma}
    \label{lem:finish_general_arb}
    Given a graph $G=(V,E)$ with max degree $\Delta\leq O(\alpha^{32})$ we can compute a $2$-ruling set in $\widetilde{O}(\log^{5/8} \alpha) + \widetilde{O}(\log^{5/3} \log n)$ rounds w.h.p.\ 
\end{lemma}

\begin{proof}
    Let $f$ be a parameter that will be fixed later during this proof.
    First we run the sparsification procedure of~\Cref{lem:spars} to obtain a vertex-subset $S \subseteq V(G)$ in $O(\log_f \Delta)$ rounds such that in $\Delta(G[S])=O(f \cdot \log n)$ and every vertex in $V$ is either in $S$ or has a neighbor in $S$. In order to obtain a $2$-ruling set of $G$, it suffices to compute an MIS on $G[S]$. We run Ghaffari's MIS algorithm~\cite{GhaffariImproved16}, which computes an MIS in $O(\log \Delta(G[S]))=O(\log(f \cdot \log n))=O(\log f)+O(\log \log n)$ rounds, such that the remaining graph consists of components of size $O(\Delta(G[S])^4 \cdot \log_\Delta n)= O((f \cdot \log n )^4\cdot \log_\Delta n )$. Finally, we invoke the MIS algorithm of~\cite{GG24} to compute an MIS on each of the remaining components in~$\widetilde{O}(\log^{5/3} N)=\widetilde{O}(\log^{5/3} ((f \cdot \log n)^4 \cdot \log_\Delta n))= \widetilde{O}(\log^{5/3} f + \log^{5/3} \log n)$. So in total, this requires
    \begin{align*}
        O(\log_f \Delta)+O(\log f)+\widetilde{O}(\log^{5/3} f )+ \widetilde{O}(\log^{5/3} \log n) \\ = O\left(\frac{\log \alpha}{\log f}\right)+\widetilde{O}(\log^{5/3} f)+\widetilde{O}(\log^{5/3} \log n)
    \end{align*}
    rounds, where we used that $\Delta\leq \alpha^{32}$. Finally, choosing $f=\exp(\log^{3/8} \alpha)$ we obtain a runtime of
    \begin{align*}
       & O\left(\frac{\log \alpha}{\log f}\right)+\widetilde{O}(\log^{5/3} f)+\widetilde{O}(\log^{5/3} \log n) \\  &= O\left(\frac{\log \alpha}{\log^{3/8}\alpha}\right)+\widetilde{O}((\log^{3/8} \alpha)^{5/3})+ \widetilde{O}(\log^{5/3} \log n)\\&
          = \widetilde{O}(\log^{5/8} \alpha+\log^{5/3}\log n)
    \end{align*}
    rounds.
\end{proof}

Finally, we are ready to prove~\Cref{thm:MORE}.

\begin{proof}[Proof of~\Cref{thm:MORE}]
Consider the following algorithm:
    Invoke the procedure of~\Cref{cor:degreeDropArb}/\Cref{cor:degree_drop} (they are the same) to compute an independent set $S$. Now consider the residual graph $G'\coloneq G\setminus(S \cup N_2(S))$.
    
    If $\alpha=O(\poly \log n)$ we know that the max degree of $G'$ is $\Delta'=O(\poly \log n)$ after $O(\log \log \Delta)$ rounds by~\Cref{cor:degree_drop} and thus can invoke the deterministic MIS algorithm of Ghaffari and Grunau~\cite{GG24} to obtain an MIS on the remaining graph in $\widetilde{O}(\log^{5/3} \Delta')=\widetilde{O}(\log^{5/3} \log n)$ rounds, yielding an overall runtime of $O(\log \log \Delta)+\widetilde{O}(\log^{5/3} \log n) = \widetilde{O}(\log^{5/3} \log n)$. We have already established correctness, independence, and domination for the degree drop, and since the MIS computation is correct, we in fact compute a $2$-ruling set.
    
    Otherwise, if $\alpha=\Omega(\poly \log n)$, we invoke the procedure of~\Cref{lem:finish_general_arb} after the degree drop. 

    \textbf{Correctness of the Subroutines.} 
    After the invocation of the degree drop procedure, the max degree dropped to $O(\alpha^{32})$ w.h.p.\, by~\Cref{cor:degreeDropArb}. Thus, we can call the procedure of~\Cref{lem:finish_general_arb} to compute a $2$-ruling set on the remaining graph.

    \textbf{Runtime.} Reducing the max degree down to $\alpha^{32}$ takes $O(\log \log \Delta)$ rounds, by~\Cref{cor:degreeDropArb}. The computation of the $2$-ruling set on the remaining graph take $\widetilde{O}(\log^{5/8} \alpha+\log^{5/3} \log n)$ by~\Cref{lem:finish_general_arb}.

    \textbf{Independence and Domination.} In each called procedure, we always compute independent sets and remove each computed independent set and its neighborhood before computing the next independent set, ensuring independence across the two procedures. 
    Before the final computation of a $2$-ruling set on the remaining graph, we only removed nodes that are dominated by distance at most $2$, thus $S\cup I$ is a $2$-ruling set.
\end{proof}

For a restricted regime of $\alpha$~\Cref{thm:MORE} runs in $\widetilde{O}(\log^{5/3} \log n)$ rounds.

\corARB*


\begin{remark*}
    In order to make~\Cref{thm:MORE} work in the \CONGEST model, one could use one of the recent network decomposition algorithms, e.g., \cite{RG20,GGHIR23} to compute an MIS on the remaining components. This would result in a larger $O(\poly \log \log n)$ term in the final runtime.
\end{remark*}

\section{Massively Parallel 2-Ruling Set}

In this section, the words node and machine are interchangeable, since each node is hosted on a unique machine with local space $O(n^\eps)$ with $\eps \in (0,1)$. The algorithm relies on \Cref{lem:degree_drop_iteration} in order to reduce the maximum degree until every node can gather their $\poly(\log \log n)$-hop neighborhood and simulate \Cref{thm:MORE} directly. Gathering neighborhoods is done via graph exponentiation \cite{wattenhofer2010}, which is a standard tool in \mpc. 

\thmRSloglogArbMPC*

\begin{proof}
	By \Cref{lem:degree_drop_iteration}, there exists a constant-round \LOCAL algorithm that computes an independent set $S$ on a graph $G$ with arboricity $\alpha$ and maximum degree $\Delta \geq \max\{(82\alpha^4)^8,\log^8 n\}$, such that the maximum degree of $G \setminus \{S \cup N_2(S)\}$ is $\Delta' \leq \Delta^{3/4}$. This algorithm can be executed directly in low-space \mpc within the same runtime: nodes with degree $<n^\eps$ can execute it locally, and nodes with degree $>n^\eps$ can execute it via an $n^\eps$-ary broadcast tree, which is straightforward to implement as $v$ only needs to aggregate the minimum real computed by its active neighbors.
	
	By applying \Cref{lem:degree_drop_iteration} $O(\log \log \log n)$ times on the input graph, the remaining graph has maximum degree 
	$$\Delta' \leq \max\{\Delta^{1/\poly(\log \log n)},(82\alpha^4)^8,\log^8 n\}~.$$ Now, each node gathers its $\poly(\log \log n)$-hop neighborhood using graph exponentiation in $O(\log \log \log n)$ rounds. This is feasible, because the size of any such neighborhood is bounded: since $\Delta \leq n$ and $\alpha \leq 2^{\poly(\log \log n)}$, it holds that 
	\begin{align*}
		(\Delta')^{\poly(\log \log n)} &\leq (\max\{\Delta^{1/\poly(\log \log n)},(66\alpha^4)^8,\log^8 n\})^{\poly(\log \log n)} \\
		&\leq (\max\{n^{1/\poly(\log \log n)},2^{\poly(\log \log n)}, \log^8 n)\})^{\poly(\log \log n)} \\
		&= \max\{n^{\eps/2}, 2^{\poly(\log \log n)}, 2^{\poly(\log \log n)} \} \\
		&= O(n^{\eps/2})~.
	\end{align*}
	
	The exponent of $n$ can be set to $\eps/2$ by choosing a suitable constant in $1/\poly(\log \log n)$, i.e., how many times we execute the degree reduction. Hence, the total number of edges in every $\poly(\log \log n)$-hop neighborhood is bounded by $O(n^\eps)$ and the total space of the algorithm is bounded by $O(m+n^{1+\eps})$. Now, we can execute any $\poly(\log \log n)$-round \LOCAL algorithm in constant time due to the fact that any $t$-round \LOCAL algorithm is simply a function from the ball of radius $t$ around every node $v$, to an output set. In particular, we can execute \Cref{thm:MORE} when arboricity $\alpha \leq 2^{\poly(\log \log n)}$, since the runtime is
	\begin{align*}
		\widetilde{O}(\log^{5/8} \alpha+\log^{5/3} \log n) &= O(\log \log \alpha \cdot \log^{5/8} \alpha + \log \log \log n \cdot \log^{5/3} \log n) \\
		&= O(\log \log \log n \cdot \poly(\log \log n) + \log \log \log n \cdot \log^{5/3} \log n) \\
		&= \poly(\log \log n)~,
	\end{align*}
	
	yielding a 2-ruling set for the current graph. Combining this set with the independent sets computed during the degree reduction (repeated application of \Cref{lem:degree_drop_iteration}) gives a 2-ruling set for the input graph by the correctness proof of \Cref{thm:MORE}.
\end{proof}

\section*{Acknowledgements}
We thank Lasse Leskelä for helpful discussions on probabilistic arguments for handling dependencies in an earlier version of this work.

\newpage
\printbibliography[heading=bibintoc]

\newpage
\appendix
\allowdisplaybreaks
\section{Extended Related Work}
\label[appendix]{apx:related_work}

\begin{table}[!h]
    \centering
    \renewcommand{\arraystretch}{1.65}
    \begin{tabular}{|c|c|c|l|}
    \hline
         Domination $\beta$ & Running Time & Graph Class & Citation  \\ \hline

         \multirow{7}{*}{$\beta=1$ (MIS)}
         & \cellcolor{lowerred}$\Omega(\min\{\log \Delta,\sqrt{\log n} \})$ 
         & Line Graphs of trees
         & \cite{khoury2025roundeliminationselfreductionclosing}  \\ \hhline{|~|-|-|-|}

        & \cellcolor{lowerred}$\Omega(\min\{\sqrt{\frac{\log n}{\log \log n}},\frac{\log \Delta}{\log \log \Delta}\})$ 
        & \multirow{2}{*}{General Graphs}
        & \cite{10.1145/3519270.3538419} \\  \hhline{|~|-|~|-|}
         
        & \cellcolor{upperblue}$O(\log \Delta)+O(\poly \log \log n)$ 
        & 
        & \cite{GhaffariImproved16} \\ \hhline{|~|-|-|-|}
        
        & \cellcolor{upperblue}$O(\sqrt{\frac{\log n}{\log \log^* n}})$ 
        & Trees
        & \multirow{2}{*}{\cite{khoury2025breakingbarriersdistributedmis}} \\   \hhline{|~|-|-|~|} 
        & \cellcolor{upperblue}$O({\frac{\log \Delta}{\log \log^* \Delta}}+\poly \log \log n)$ 
        & Girth at least $7$
        &  \\  \hhline{|~|-|-|-|}

       & \cellcolor{upperblue}$O(\log\alpha + \sqrt{\log n})$ 
       & Arboricity $\alpha$
       & \multirow{2}{*}{\cite{BEPSv3}+\cite{GhaffariImproved16}} \\  \hhline{|~|-|-|~|}
         & \cellcolor{upperblue}$O(\sqrt{\log n})$ 
         & Arboricity  $\leq 2^{\sqrt{\log n}}$
         &  \\ \hline 

        $\beta \leq c \cdot M_R^1(\Delta,n)$
        & \cellcolor{lowerred}$\Omega(\min\{ \frac{\log \Delta}{\beta \log \log \Delta},\log_\Delta \log n\})$ 
        & \multirow{5}{*}{General Graphs}
        & \multirow{2}{*}{\cite{BBO22}} \\  \hhline{|-|-|~|~|}

        $\beta \leq c\cdot \sqrt[3]{\frac{\log \log n}{\log \log \log n}}$
        & \cellcolor{lowerred}$\Omega(\sqrt{ \frac{\log\log n}{\beta \log \log \log n}})$ 
        & 
        &  \\  \hhline{|-|-|~|-|}

        $\beta \leq \epsilon \cdot M_R^2(\Delta,n)$
        & \cellcolor{lowerred}$\Omega(\min\{\beta \Delta^{1/\beta},\log_\Delta \log n\})$ 
        &
        & \multirow{2}{*}{\cite{BBKO2021hideandseek}} \\  \hhline{|-|-|~|~|}

         $\beta \leq \epsilon \cdot \sqrt{\frac{\log \log n}{\log \log \log n}}$
         & \cellcolor{lowerred}$\Omega(\frac{\log \log n}{\beta \log \log \log n})$ 
         &
         & \\  \hhline{|-|-|~|-|}
         
        \multirow{8}{*}{$\beta=2$}
        & \cellcolor{upperblue}$O(\sqrt{\log \Delta}+\poly \log \log n)$ 
        &
        & \cite{tushar-super-fast-2014}+\cite{GhaffariImproved16} \\  \hhline{|~|-|-|-|}

        & \cellcolor{upperblue}$O(\log \log n)$ 
        & Trees
        & \multirow{2}{*}{\cite{BMU25}} \\ \hhline{|~|-|-|~|}

        & \cellcolor{upperblue}$\widetilde{O}(\log^{5/3} \log n)$ 
        & Girth at least $7$
        &  \\ \hhline{|~|-|-|-|}

        & \cellcolor{upperblue}$O\Big(\log^{1/4} \Delta +\log \alpha+\poly \log \log n\Big)$ 
        & Arboricity $\alpha$
        & \multirow{2}{*}{%
  \shortstack{
    \cite{BEPSv3}+\cite{tushar-super-fast-2014}\\
    +\cite{GhaffariImproved16}
  }%
} \\ \hhline{|~|-|-|~|}

         & \cellcolor{upperblue}$O\Big(\log^{1/4} \Delta +\poly \log \log n\Big)$ 
        & Arboricity $ O(1)$
        &  \\ \hhline{|~|-|-|-|}
        
       & \cellcolor{upperblue}$O(\log \log n)$ 
       & Arboricity $\leq O(\log \log n)$
       & \textbf{\Cref{thm:2RSconstArb}} \\ \hhline{|~|-|-|-|}

         & \cellcolor{upperblue}$O(\log^{5/8} \alpha)+\widetilde{O}(\log^{5/3} \log n)$ 
         & Arboricity $\alpha$
         & \textbf{\Cref{thm:MORE}} \\ \hhline{|~|-|-|-|}

         & \cellcolor{upperblue}$\widetilde{O}(\log^{5/3} \log n)$ 
         & Arboricity $\leq 2^{(\log \log n)^{8/3}}$
         & \textbf{\Cref{cor:ARB}} \\ \hline
         
        $\beta =3 $
        & \cellcolor{upperblue}$O(\log^3\log n)$ 
        & Arboricity $ O(1)$
        & \cite{kothapalli-superfast2012} \\ \hline
         
         $\beta = O(\log \log \log n)$
         & \cellcolor{upperblue}$\widetilde{O}(\log \log n) $ 
         & Girth at least $7$
         & \cite{BMU25} \\ \hline 
         
         $\beta= O(\log \log n)$
         & \cellcolor{upperblue}$O(\log \log n)$ 
         & General Graphs
         & \cite{SEW13}+\cite{Gfeller07} \\ \hline  

        $\beta \geq 2$
        & \cellcolor{upperblue}$O(\beta \cdot \log^{1/\beta} \Delta)+\poly \log \log n$ 
        & General Graphs
        & \cite{tushar-super-fast-2014} + \cite{GhaffariImproved16} \\ \hline 
    \end{tabular}

    \caption[Caption for LOF]{Randomized Ruling Sets. Define $M_R^1(\Delta,n) \coloneq \min\{\sqrt{\frac{\log \Delta}{\log \log \Delta}},\log_\Delta \log  n\}$ and $M_R^2(\Delta,n) \coloneq  \min\{\log \Delta,\log_\Delta \log  n\}$. All algorithms are randomized, and lower bounds hold for randomized algorithms.} 
    \label{tab:rand-rul-sets-extended}
\end{table}

\clearpage

\begin{table}[!h]
    \centering
    \begin{tabular}{|c|c|c|l|}
    \hline
         Domination $\beta$ & Running Time & Graph Class & Citation  \\ \hline

        \multirow{3}{*}{$\beta=1$ (MIS)}
        & \cellcolor{upperblue}$\widetilde{O}(\log^{5/3}n)$  
        & General Graphs 
        & \cite{GG24} \\ \hhline{|~|-|-|-|}

        & \cellcolor{upperblue}$O(\frac{\log n}{\log \log n})$ 
        & Arboricity of $O(\log^{1/2-\delta} n)$
        & \multirow{2}{*}{\cite{Barenboim2010}} \\ \hhline{|~|-|-|~|}

        & \cellcolor{upperblue}$O(\alpha  \sqrt{\log n}+\alpha \cdot \log \alpha)$ 
        & Arboricity of $\Omega(\sqrt{\log n})$
        & \\ \hline

        \multirow{2}{*}{$\beta=2$}
        & \cellcolor{upperblue}
        & \multirow{2}{*}{Arboricity $\alpha$}
        & \cite{Barenboim2010}+\cite{KMW18-rulingSets} \\ \cline{4-4}

        & \cellcolor{upperblue}\multirow{-2}{*}[0pt]{$O(\log n + \alpha)$}
        & 
        & \textbf{\Cref{thm:det2RS}} \\ \hline

         $\beta=O(1)$ 
         & \cellcolor{upperblue}$O( \Delta^{2/(\beta+2)}+\log^* n)$ 
         & General Graphs
         & \cite{M21} \\ \hline

        $\beta \leq c \cdot M_D^1(\Delta,n)$ 
        & \cellcolor{lowerred}$\Omega(\min\{ \frac{\log \Delta}{\beta \log \log \Delta},\log_\Delta n\})$ 
        & \multirow{6}{*}{General Graphs}
        & \multirow{2}{*}{\cite{BBO22}} \\ \hhline{|-|-|~|~|}

        $\beta \leq c\cdot \sqrt[3]{\frac{\log n}{\log \log n}}$ 
        & \cellcolor{lowerred}$\Omega(\sqrt{ \frac{\log n}{\beta \log \log n}})$ 
        & 
        & \\ \hhline{|-|-|~|-|}

        $\beta \leq \epsilon \cdot M_D^2(\Delta,n)$ 
        & \cellcolor{lowerred}$\Omega(\min\{\beta \Delta^{1/\beta},\log_\Delta n\})$ 
        & 
        & \multirow{2}{*}{\cite{BBKO2021hideandseek}} \\ \hhline{|-|-|~|~|}

         $\beta \leq \epsilon \cdot \sqrt{\frac{\log n}{\log \log n}}$ 
         & \cellcolor{lowerred}$\Omega(\frac{\log n}{\beta \log \log n})$ 
         & 
         & \\ \hhline{|-|-|~|-|}
         
         $\beta = O(\log \log n)$ 
         & \cellcolor{upperblue}$\widetilde{O}(\log n)$  
         & 
         & \cite{GG24} \\ \hhline{|-|-|~|-|}

         $\beta = O(\log n)$ 
         & \cellcolor{upperblue}$O(\log n)$  
         & 
         & \cite{awerbuch89} \\ \hline

    \end{tabular}
    \caption[Caption for LOF]{Deterministic Ruling Sets. 
    $M_D^1(\Delta,n) \coloneqq 
    \min\{\sqrt{\frac{\log \Delta}{\log \log \Delta}},\log_\Delta n\}$ 
    and 
    $M_D^2(\Delta,n) \coloneqq 
    \min\{\log \Delta,\log_\Delta n\}$. All algorithms are deterministic, and lower bounds hold for deterministic algorithms.}
    \label{tab:det-rul-sets}
\end{table}

For completeness, we show how to combine known results of~\cite{tushar-super-fast-2014},\cite{BEPSv3} and \cite{GhaffariImproved16} to obtain a $2$-ruling set algorithm that is not explicitly stated in the literature.

\begin{theorem}
\label{thm:related_work}
    There is an $O(\log^{1/4} \Delta +\log \alpha+ \poly \log \log n)$ round \LOCAL algorithm to compute a $2$-ruling set in graphs with arboricity $\alpha$ w.h.p.\
\end{theorem}
\begin{proof}
    Observe if $\exp(\log^{1/4}\Delta) \leq \poly \log n$, we can compute a $2$-ruling set in $O(\poly \log \log n)$ rounds by~\cite{GhaffariImproved16}. Thus, we assume that $\exp(\log^{1/4}\Delta) \geq \poly \log n$.
    First, we use \cite[Theorem 7.2]{BEPSv3} to compute an independent set $M$ with $t=O(\exp(\log^{3/4}\Delta) $ such that the residual graph $G'=G[V(G)\setminus (M \cup N(M))]$ has maximum degree at most $O(t\alpha)$ in $O(\log _t \Delta)=O(\log^{1/4} \Delta)$ rounds. Secondly, we apply the sampling result~\Cref{lem:spars} by~\cite{tushar-super-fast-2014} with $f=\exp(\sqrt{\log n})$ to sample a subset $S\subseteq V(G')$ in $O(\log_f t\alpha)=O(\log_ft+\log_f\alpha)=O(\sqrt{\log \Delta}/\log^{1/4} \Delta+\log \alpha)=O(\log^{1/4}\Delta+\log \alpha)$ rounds such that $\Delta(G'[S])=O(f\cdot \log n)$, and every vertex in $V$ is either in $S$ or has a neighbor in $S$. 

    To compute a $2$-ruling set of $G$, it suffices to compute an MIS of $G'[S]$. For that, we apply \cite[Theorem 7.2]{BEPSv3} a second time with $t'=\exp(\log^{1/4}\Delta)$, to compute an independent set $M_S$ of $S$ in $O(\log_{t'} (f \log n))=O(\log (\exp(\sqrt{\log \Delta}) \log n)/\log^{1/4}\Delta)=O(\log^{1/4}\Delta + \poly \log \log n)$ rounds such that the residual graph $\widetilde{G}\coloneq G'[S\setminus (M_S \cup N(M_S))]$ has maximum degree at most $O(t'\alpha)$. Finally, we compute an MIS on $\widetilde{G}$ using Ghaffaris MIS algorithm~\cite{GhaffariImproved16} in $O(\log t'\alpha)+(\poly \log \log)=O(\log^{1/4}\Delta+\log \alpha +\poly \log \log n)$ rounds. 

    Overall, we compute a $2$-ruling set in $
    O(\log^{1/4}\Delta + \log \alpha+ \poly \log \log n).$
\end{proof}
If the arboricity is constant, we immediately obtain the following.
\begin{corollary}
    \label{cor:related_work}
    There is an $O(\log^{1/4} \Delta + \poly \log \log n)$ round \LOCAL algorithm to compute a $2$-ruling set in graphs with arboricity $\alpha=O(1)$ w.h.p.\
\end{corollary}

\section{Example: Successful Nodes are Correlated}
\label[appendix]{apx:counterexample}

\label{ex:independent-set-correlated-ranks}
Consider the graph in Figure~\ref{fig:independent-set-correlated-ranks}. The set $M=\{u,v\}$ is independent. Let
\(r(u),r(v),r(x)\) be mutually independent random variables, each uniformly
distributed on \([0,1]\).
\begin{figure}[h]
    \centering
    \begin{tikzpicture}[
        vertex/.style={circle, draw, minimum size=8mm, inner sep=0pt},
        every edge/.style={draw, -{Stealth[length=2mm]}, thick},
        node distance=2.2cm
    ]
        \node[vertex] (u) at (0,0) {\(u\)};
        \node[vertex] (v) at (2.2,0) {\(v\)};
        \node[vertex] (x) at (1.1,-1.6) {\(x\)};

        \draw[->, thick] (u) -- (x);
        \draw[->, thick] (v) -- (x);

    \end{tikzpicture}
    \caption{
        A three-vertex example with independent candidate set
        \(M=\{u,v\}\). Although \(u\) and \(v\) are non-adjacent, conditioning on both candidates
        being successful, i.e., on \(r(u)<r(x)\) and \(r(v)<r(x)\), creates
        correlation between the ranks \(r(u)\) and \(r(v)\).
    }
    \label{fig:independent-set-correlated-ranks}
\end{figure}

The event that both vertices in \(M\) are successful is $ A=\{r(u)<r(x),\ r(v)<r(x)\}.$
We show that \(r(u)\) and \(r(v)\) are not independent after conditioning
on \(A\).

For fixed \(t\in[0,1]\), conditioning on \(r(x)=t\) gives
\[
    \P[A\mid r(x)=t]
    =
    \P[r(u)<t,\ r(v)<t] =  \P[r(u)<t]\P[r(v)<t]
    =
    t^2,
\]
since \(r(u)\) and \(r(v)\) are independent uniform random variables on
\([0,1]\).

Note that the density function of $r(x)$ is given by $f_{r(x)}(t)=1$ for $0 \leq t \leq 1$, thus by the law of total probability,
\[
    \P[A]
    =
    \int_0^1 \P[A\mid r(x)=t]f_{r(x)}(t)\,dt
    =
    \int_0^1 t^2\,dt
    =
    \frac{1}{3}.
\]
Using Bayes' rule for densities, the conditional density of \(r(x)\)
given \(A\) is
\[
    f_{r(x)\mid A}(t)
    =
    \frac{\P[A\mid r(x)=t] f_{r(x)}(t)}{\P[A]}
    =
    \frac{t^2}{1/3}
    =
    3t^2,
    \qquad 0\le t\le 1.
\]

Conditioned on \(r(x)=t\) and on \(A\), the variables \(r(u)\) and
\(r(v)\) are independent and uniformly distributed on \([0,t]\), see~\cite{ILIOPOULOS20091008}. Hence
\[
    \E[r(u)\mid r(x)=t,A]
    =
    \frac{t}{2}.
\]
By the law of total expectation,
$$
    \E[r(u)\mid A]
    =
    \int_0^1
        \E[r(u)\mid r(x)=t,A]\,
        f_{r(x)\mid A}(t)
    =
    \int_0^1 \frac{t}{2}\cdot 3t^2\,dt
    =
    \frac{3}{2}\int_0^1 t^3\,dt
    =
    \frac{3}{8}.
$$
By symmetry, $ \E[r(v)\mid A]
    = \frac{3}{8}.$

Again, since conditioned on \(r(x)=t\) and \(A\), the variables \(r(u)\)
and \(r(v)\) are independent uniforms on \([0,t]\), we get
\[
    \E[r(u)r(v)\mid r(x)=t,A]
    =
    \E[r(u)\mid r(x)=t,A]\,
    \E[r(v)\mid r(x)=t,A]
    =
    \frac{t^2}{4}.
\]
By the law of total expectation we obtain
$$
    \E[r(u)r(v)\mid A]
    =
    \int_0^1
        \E[r(u)r(v)\mid r(x)=t,A]\,
        f_{r(x)\mid A}(t)
    \,dt
    =
    \int_0^1 \frac{t^2}{4}\cdot 3t^2\,dt
    =
    \frac{3}{4}\int_0^1 t^4\,dt
    =
    \frac{3}{20}.
$$

Therefore,
\[
    \E[r(u)r(v)\mid A]
    = \frac{3}{20} > \frac{9}{64}=
    \E[r(u)\mid A]\E[r(v)\mid A].
\]
This implies that the conditional ranks are not independent. In fact, $r(v)$ and $r(u)$ are positively correlated after conditioning on $M$ being successful,
\[
    \operatorname{Cov}(r(u),r(v)\mid A)
    =
    \frac{3}{20}
    -
    \frac{9}{64}
    =
    \frac{3}{320}
    >
    0.
\]

\section{Deferred Proofs}
\label[appendix]{apx:more_thm}

First we state the theorems from~\cite{https://doi.org/10.1002/rsa.20532} we use to derive~\Cref{lem:read-k-chernoff} and~\Cref{lem:read-k-all-p}.

\begin{theorem}[Theorem 1.1]\cite{https://doi.org/10.1002/rsa.20532}
   Let $Y_1, \dots, Y_m$ be a family of read-$k$ indicator random variables with $\Pr[Y_i=1]=p_i.$ Let $p \coloneq \frac{1}{m}\sum_{i=1}^mp_i$ and $Y \coloneq \sum_{i=1}^mY_i$. Then for any $\varepsilon >0$,
  $$
  \Pr[Y \leq (p - \varepsilon)m] \leq \exp(-D(p-\varepsilon||p)\cdot m/k).
  $$ 
  Here $D(q||p) \coloneqq q \log \left(\frac{q}{p}\right)+(1-q)\log\left(\frac{1-q}{1-p}\right)$ is the Kullback-Leibler divergence.
\end{theorem}

\begin{theorem}[Theorem 1.2]\cite{https://doi.org/10.1002/rsa.20532}
     Let $Y_1, \dots, Y_m$ be a family of read-$k$ indicator random variables with $\Pr[Y_i=1]\leq p$. Then 
    $$
    \Pr[Y_1=\dots=Y_m=1] = p^{m/k}.
    $$
\end{theorem}

\label{apx:deferred_proofs}
\begin{proof}[Proof of \Cref{lem:read-k-chernoff}]
Let $\mu \coloneq \E[Y]=\sum_{i=1}^m p_i=mp$. If $p=0$, then
$\mu=0$ and the claimed bound is trivial. Moreover, if $\delta>1$, then
$(1-\delta)\mu<0$ whenever $\mu>0$, and hence the event
$\{Y\leq (1-\delta)\mu\}$ is empty. Thus, it suffices to consider
$0<\delta\leq 1$ and $p>0$.

Set $\varepsilon \coloneq \delta p.$
Then $(p-\varepsilon)m=(1-\delta)pm=(1-\delta)\mu.$
Applying the lower-tail bound of
\cite[Theorem~1.1]{https://doi.org/10.1002/rsa.20532} gives
\[
    \Pr\!\left[Y\leq (1-\delta)\mu\right]
    \leq
    \exp \left(-D (( 1-\delta)p || p ) \cdot m/k\right).
\]
It remains to bound the binary relative entropy. Writing
$q=(1-\delta)p$, we have
\begin{align*}
D(q\|p)
&=
(1-\delta)p\log(1-\delta)
+
\bigl(1-(1-\delta)p\bigr)
\log\!\left(
    \frac{1-(1-\delta)p}{1-p}
\right) \\
&=
(1-\delta)p\log(1-\delta)
+
\bigl(1-p+\delta p\bigr)
\log\!\left(
    1+\frac{\delta p}{1-p}
\right).
\end{align*}
Using $\log(1+x)\geq x/(1+x)$ for $x\geq 0$, the second term is at
least $\delta p$. Therefore
\[
    D\!\left((1-\delta)p\,\middle\|\,p\right)
    \geq
    p\bigl((1-\delta)\log(1-\delta)+\delta\bigr).
\]
Finally, for $0\leq \delta\leq 1$,we have $(1-\delta)\log(1-\delta)+\delta \geq \frac{\delta^2}{2}.$
Indeed, if $
    h(\delta)
    \coloneq
    (1-\delta)\log(1-\delta)+\delta-\frac{\delta^2}{2},$
then $h(0)=0$ and
$
    h'(\delta)=-\log(1-\delta)-\delta\geq 0.
$
Consequently,
$    D\!\left((1-\delta)p\,\middle\|\,p\right)
    \geq \frac{\delta^2p}{2}.
$
Substituting this into the preceding tail bound and using $mp=\mu$
yields
\[
    \Pr\!\left[Y\leq (1-\delta)\E[Y]\right]
    \leq
    \exp\!\left(
        -\frac{\delta^2 mp}{2k}
    \right)
    =
    \exp\!\left(
        -\frac{\delta^2\E[Y]}{2k}
    \right).
\]
\end{proof}

\begin{proof}[Proof of \Cref{lem:read-k-all-p}]
Let $p_i \coloneq \Pr[Y_i=1]$, so that $p_i\leq p$ for every $i\in[m]$. If $p=1$, the claimed
inequality is immediate. Assume henceforth that $p<1$.

For every $i\in[m]$, introduce an indicator random variable $Z_i$,
independent of all random variables defining the read-$k$ family and
independent of all other $Z_j$, such that
\[
    \Pr[Z_i=1]
    =
    \frac{p-p_i}{1-p_i}.
\]
This quantity lies in $[0,1]$ because $p_i\leq p<1$. Define
\[
    \widetilde{Y}_i \coloneq Y_i \lor Z_i.
\]
Then
$\Pr[\widetilde{Y}_i=1]  = p_i+(1-p_i)\Pr[Z_i=1] = p_i+(1-p_i)\frac{p-p_i}{1-p_i}= p.
$
Moreover, $\widetilde{Y}_1,\dots,\widetilde{Y}_m$ is again a
read-$k$ family: every original underlying random variable is read by
at most $k$ of the functions, and the newly introduced variable $Z_i$
is read only by $\widetilde{Y}_i$.

Since $Y_i\leq \widetilde{Y}_i$ for every $i$, we have
\[
    \{Y_1=\dots=Y_m=1\}
    \subseteq
    \{\widetilde{Y}_1=\dots=\widetilde{Y}_m=1\}.
\]
Applying \cite[Theorem~1.2]{https://doi.org/10.1002/rsa.20532} to
the read-$k$ family
$\widetilde{Y}_1,\dots,\widetilde{Y}_m$, whose marginal probabilities
are all equal to $p$, gives
\[
    \Pr[Y_1=\dots=Y_m=1]
    \leq \Pr[\widetilde{Y}_1=\dots=\widetilde{Y}_m=1]
    \leq p^{m/k}.
\]
\end{proof}

\section{Probabilistic Theorems}
\label{apx:prob-thm}

We restate some definitions and results that are contained in~\cite{EPW67} and are rewritten in a more modern way in~\cite{Zhao2022CorrelationInequalities} that we use in the proof of~\Cref{lem:node_joins}.

For mutually independent random variables $X_1,\dots, X_n$, each on a linear ordered set let $X=(X_i)_{i\in [n]}$. We say that a function $f(X)$ is non-decreasing if $f(X) \leq f(X')$ whenever $X\leq X'$ coordinate-wise, i.e., $X_i\leq X_i'$ for $1\leq i\leq n$.


\begin{theorem}\cite[Definition 1.1, Theorem 2.1]{EPW67}
\label{thm:associated}
    For independent random variables $X_1,\dots, X_n$ and $X= (X_i)_{i\in [n]}$ we have 
     $$
Cov(f(X),g(X)) \geq 0
$$
for all non decreasing functions $f$ and $g$ for which $\E[f(X)], \E[g(X)], \E[f(X)g(X)]$ exist.
\end{theorem}

\end{document}